\documentclass[aps,pra,twocolumn,nopacs]{revtex4}
\usepackage{amsthm}
\usepackage{amsmath}
\usepackage{latexsym}
\usepackage{amsfonts}
\usepackage{amssymb}
\usepackage{color}
\usepackage{bbm,dsfont}
\usepackage{graphicx}

\usepackage{subfigure}
\usepackage{mathrsfs}
\usepackage{mathbbol}
\usepackage{hyperref}
\usepackage{enumerate}
\usepackage{upgreek}

\newtheorem{proposition}{Proposition}
\newtheorem{proposition?}{Proposition?}

\newtheorem{lemma}{Lemma}

\theoremstyle{definition}

\newtheorem{remark}{Remark}
\newtheorem{example}{Example}
\newtheorem{definition}{Definition}

\definecolor{green}{RGB}{42, 111, 55}
\definecolor{blue}{RGB}{6, 7, 135}
\definecolor{red}{RGB}{188, 36, 60}

\newcommand{\real}{\mathbb R} 
\newcommand{\complex}{\mathbb C} 
\newcommand{\half}{\tfrac{1}{2}} 

\newcommand{\hi}{\mathcal{H}} 
\newcommand{\hik}{\mathcal{K}} 
\newcommand{\tr}[1]{\mathrm{tr}\left[#1\right]} 
\newcommand{\ptr}[2]{\mathrm{tr}_{#1}[#2]} 
\newcommand{\id}{\mathbbm{1}} 

\newcommand{\A}{\mathsf{A}}
\newcommand{\B}{\mathsf{B}}
\newcommand{\C}{\mathsf{C}}
\newcommand{\D}{\mathsf{D}}
\newcommand{\E}{\mathsf{E}}
\newcommand{\F}{\mathsf{F}}
\newcommand{\G}{\mathsf{G}}
\renewcommand{\P}{\mathsf{P}}
\newcommand{\M}{\mathsf{M}}

\newcommand{\vsigma}{\mathbf{\sigma}} 
\newcommand{\va}{\mathbf{a}} 
\newcommand{\vb}{\mathbf{b}} 
\renewcommand{\vr}{\mathbf{r}} 
\newcommand{\vx}{\mathbf{x}} 
\newcommand{\vy}{\mathbf{y}} 

\newcommand{\hin}{\hi_{in}} 
\newcommand{\state}{\mathcal{S}}

\usepackage{comment}
\newcommand{\vz}{\mathbf{z}}
\newcommand{\x}{\mathrm{x}}
\newcommand{\y}{\mathrm{y}}
\newcommand{\z}{\mathrm{z}}
\usepackage{xcolor}

\usepackage{multirow}

\begin{document}

\title[]{Testing incompatibility of quantum devices with few states}

\author{Teiko Heinosaari}
\email{teiko.heinosaari@utu.fi}
\address{Department of Physics and Astronomy, University of Turku, Finland}
\address{Quantum algorithms and software, VTT Technical Research Centre of Finland Ltd}

\author{Takayuki Miyadera}
\email{miyadera@nucleng.kyoto-u.ac.jp}
\address{Department of Nuclear Engineering, Kyoto University, 6158540 Kyoto, Japan}

\author{Ryo Takakura}
\email{takakura.ryo.27v@st.kyoto-u.ac.jp}
\address{Department of Nuclear Engineering, Kyoto University, 6158540 Kyoto, Japan}

\begin{abstract}
When observations must come from incompatible devices and cannot be produced by compatible devices?
This question motivates two integer valued quantifications of incompatibility, called incompatibility dimension and compatibility dimension.
The first one quantifies how many states are minimally needed to detect incompatibility if the test states are chosen carefully, whereas the second one quantifies how many states one may have to use if they are randomly chosen.
With concrete examples we show that these quantities have unexpected behaviour with respect to noise. \\
\end{abstract}


\maketitle

\section{Introduction}
Quantum information processing, including the exciting fields of quantum communication and quantum computation, is ultimately based on the fact that there are new types of resources that can be utilized in carefully designed information processing protocols. 
The best known feature of quantum information is that quantum systems can be in superposition and entangled states, and these resources lead to applications such as superdense coding and quantum teleportation.
While superposition and entanglement are attributes of quantum states,  quantum measurements have also features that can power new type of applications. 
The best known and most studied property is the incompatibility of pairs (or collections) of quantum measurements \cite{HeMiZi16}.
It is crucial e.g. in the BB84 quantum key distribution protocol \cite{BeBr84} that the used measurements are incompatible.

From the resource perspective, it is important to 
quantify the incompatibility. 
There has been several studies on incompatibility robustness, i.e., how incompatibility is affected by noise.
This is motivated by the fact that noise is unavoidable in any actual implementation of quantum devices and similar to other quantum properties (e.g. entanglement), large amount of noise destroys incompatibility.  
Earlier studies have mostly focused in quantifying noise \cite{DeFaKa19} and finding those pairs or collections of measurements that are most robust to certain types of noise \cite{HeScToZi14}, or to find conditions under which all incompatibility is completely erased \cite{HeKiReSc15}.
In this work we introduce quantifications of incompatibility which are motivated by operational aspect of testing whether a collection of devices is incompatible or not. 
We focus on two integer valued quantifications of incompatibility, called \emph{compatibility dimension} and \emph{incompatibilility dimension}. 
We formulate these concepts for arbitrary collections of devices. 
Roughly speaking, the first one quantifies how many states we minimally need to use to detect incompatibility if we choose the test states carefully, whereas the second one quantifies how many (affinely independent) states we may have to use if we cannot control their choice.
We study some of the basic properties of these quantifications of incompatibility and we present several examples to demonstrate their behaviour.

We show that, remarkably, even for the standard example of noisy orthogonal qubit observables the incompatibility dimension has a jump in a point where all noise robustness measures are continuous and indicate nothing special to happen.
More precisely, the noise parameter has a threshold value where the number of needed test states to reveal incompatibility shifts from 2 to 3. 
This means that even in this simple class of incompatible pairs of qubit observables there is a \emph{qualitative} difference in the incompatibility of less noisy and more noisy pairs of observables. 
An interesting additional fact is that the compatibility dimension of these pairs of observables does not depend on the noise parameter.

For simplicity and clarity, we will restrict to finite dimensional Hilbert spaces and observables with finite number of outcomes. 
Our definitions apply not only to quantum theory but also to any general probabilistic theory (GPT) \cite{Kuramochi20, Plavala21}.
However, for the sake of concreteness we keep the discussion in the realm of quantum theory. 
The main definitions work in any GPT without any changes.
We expect that similar findings as the aforementioned result on noisy orthogonal qubit observables can be made in subsequent studies on other collections of devices.  

Related studies have been recently reported in \cite{GuQuAo19,Kiukas20,LoNe21,UoKrDeMiTaPeGuBr21} in the case of quantum observables.
We will explain the interconnections of these studies to ours in Section \ref{sec:dim}
 once the relevant definitions have been introduced.

\section{(In)compatibility on a subset of states}

A quantum observable is mathematically described as a positive operator valued measure (POVM) \cite{MLQT12}. 
A quantum observable with finite number of outcomes is hence a map $x \mapsto \A(x)$ from the outcome set to the set of linear operators on a Hilbert space.
We recall that the compatibility of quantum observables $\A_1,\ldots,\A_n$ with outcome sets $X_1,\ldots,X_n$ means that there exists an observable $\G$, called \emph{joint observable}, defined on the product outcome set $X_1 \times \cdots \times X_n$ such that from an outcome $(x_1,\ldots,x_n)$ of $\G$, one can infer outcomes for every $\A_1,\ldots,\A_n$ by ignoring the other outcomes.
More precisely, the requirement is that
\begin{equation}\label{eq:pp}
\begin{split}
& \A_1(x_1) = \sum_{x_2,\ldots,x_n} \G(x_1,x_2,\ldots,x_n) \\
& \A_2(x_2) = \sum_{x_1,x_3\ldots,x_n} \G(x_1,x_2,\ldots,x_n) \\
& \quad\vdots \\
& \A_n(x_n) = \sum_{x_1,\ldots,x_{n-1}} \G(x_1,x_2,\ldots,x_n) \\
\end{split}
\end{equation}

If $\A_1,\ldots,\A_n$ are not compatible, then they are called \emph{incompatible}.
This definition applies in all general probabilistic theories and has, in fact, led to inspiring findings on quantum incompatibility compared to incompatibility in other general probabilistic theories \cite{BuHeScSt13,JePl17,TaMi20}. 

\begin{example}(\emph{Unbiased qubit observables})\label{ex:qubit-1}
We recall a standard example to fix the notation that we will use in later examples.
An unbiased qubit observable is a dichotomic observable with outcomes $\pm$ and
determined by a vector $\va\in\real^3$, $|\va| \leq 1$ via 
\begin{equation*}
\A^{\va}(\pm) = \half ( \id \pm \va \cdot \vsigma ) \,,
\end{equation*}
where $\va \cdot \vsigma = a_1\sigma_1 + a_2\sigma_2 + a_3\sigma_3$ and $\sigma_i$, $i=1,2,3$, are the Pauli matrices.
The Euclidean norm $|\va|$ of $\va$ reflects the noise in $\A^{\va}$; in the extreme case of $|\va|=1$ the operators $\A^{\va}(\pm)$ are projections and the observable is called sharp.
As shown in \cite{Busch86}, two unbiased qubit observables $\A^{\va}$ and $\A^{\vb}$ are compatible if and only if 
\begin{equation}\label{eq:paul}
|\va + \vb| + |\va - \vb| \leq 2 \, .
\end{equation}
There are two extreme cases.
Firstly, if $\A^{\va}$ is sharp then it is compatible with some $\A^{\vb}$ if and only if $\vb=r\va$ for some $-1\leq r \leq 1$.
Secondly, if $|\va|=0$, then $\A^{\va}(\pm) = \half \id$ and it is called a trivial qubit observable, in which case it is compatible with all other qubit observables. 
\end{example}

How can we test if a given family of observables is compatible or incompatible?
From the operational point of view, the existence of 
an observable $\G$ satisfying 
(\ref{eq:pp}) is equivalent to 
the existence of $\G$ such that for any state $\varrho$ the equation
\begin{equation}\label{eq:naive}
\tr{\varrho \A_1(x_1)}=  \sum_{x_2,\ldots,x_n} \tr{\varrho \G(x_1,x_2,\ldots,x_n)} 
\end{equation}
holds.
To test the incompatibility we should hence check the validity of \eqref{eq:naive} in a subset of states that spans the whole state space.
An obvious question is then if we really need all those states, or if a smaller number of test states is enough.
Further, does the number of needed test states depend on the given family of observables? How does noise affect the number of needed test states?

Before contemplating into these questions, we recall that analogous definitions of compatibility and incompatibility make sense for other types of devices, in particular, for instruments and channels \cite{HeMiRe14,Haapasalo15,HeMiZi16,HeMi17,HeReRyZi18, Kuramochi18a, Haapasalo19}. 
We limit our discussion to quantum devices although, again, the definitions apply to devices in general probabilistic theories.
We denote by $\state(\hi)$ the set of all density operators on a Hilbert space $\hi$.
The input space of all types of devices must be $\state(\hin)$ on the same Hilbert space $\hin$
as the devices operate on a same system.
We denote $\state(\hin)$ simply by $\state$.
A device is a completely positive map and the `type' of the device is 
characterised by its output space.
Output spaces for the three basic types of devices are:
\begin{itemize}
\item observable: $P(X):=\{p=\{p(x)\}_{x\in X}\mid0 \leq p(x)\leq 1,\ 
\sum_{x}p(x)=1\}$ ,
\item channel: $\mathcal{S}(\hi_{out})$ , 
\item instrument: $\mathcal{S}(\hi_{out})\otimes 
P(X)$\, . 
\end{itemize}
In this classification an observable $\A$ is identified with a map $\varrho\mapsto\tr{\varrho \A(x)}$ from $\state(\hin)$ to $P(X)$. 
We limit our investigation to the cases where the number of outcomes in $X$ is finite and the output Hilbert space $\hi_{out}$ is finite dimensional.
Regarding $P(X) \subset \state(\complex^{|X|})$ as the set of all diagonal density operators, we can summarize that quantum devices are normalized completely positive maps to different type of output spaces.  

Devices $\D_1, \ldots, \D_n$ are \emph{compatible} if there exists a device $\D$ that can simulate $\D_1, \ldots, \D_n$ simultaneously, meaning that by ignoring disjoint parts of the output of $\D$ we get the same actions as $\D_1, \ldots, \D_n$ (see \cite{HeMiZi16}).  
This kind of device is called a \emph{joint device} of $\D_1, \ldots, \D_n$. 
The input space of $\D$ is the same as for $\D_1, \ldots, \D_n$, but the output space is the tensor product of their output spaces.
As an illustration, let $\D_j\colon\state(\hi_{in})\to\state(\hi_j)$ $(j=1, \ldots, n)$ be quantum channels. 
They are compatible iff there exists a channel $\D\colon\state(\hi_{in})\to\state(\bigotimes_{j=1}^{n}\hi_j)$ satisfying
\begin{equation*}
	\begin{split}
		& \D_1(\varrho) = \mathrm{tr}_{\hi_2,\ldots,\hi_n} \D(\varrho) \\
		& \D_2(\varrho) = \mathrm{tr}_{\hi_1, \hi_3, \ldots,\hi_n} \D(\varrho) \\
		& \quad\vdots \\
		& \D_n(\varrho) = \mathrm{tr}_{\hi_1,\ldots,\hi_{n-1}} \D(\varrho) \\
	\end{split}
\end{equation*}
for all $\varrho\in\state(\hi_{in})$ (see \eqref{eq:pp}).
If $\D_1, \ldots, \D_n$ are not compatible, then they are \emph{incompatible}.

We recall a qubit example to exemplify the general definition. 

\begin{example}(\emph{Unbiased qubit observable and partially depolarizing noise})\label{ex:pauli}
A measurement of an unbiased qubit observable $\A^{\va}$ necessarily disturbs the system.
This trade-off is mathematically described by the compatibility relation between observables and channels. 
Let us consider partially depolarizing qubit channels, which have the form
\begin{equation}
\Gamma_{p}(\varrho) = p \varrho + (1-p) \tfrac{1}{2} \id 
\end{equation}
for $0 \leq p \leq 1$.
A joint device for a channel and observable is an instrument. 
Hence, $\A^{\va}$ and $\Gamma_p$ are compatible if there exists an instrument $x \mapsto \Phi_x$ such that
\begin{equation*}
\sum_x \Phi_x(\varrho) = \Gamma_p(\varrho) \quad \textrm{and} \quad \tr{\Phi_x(\varrho)} = \tr{\varrho \A^{\va}(x)}
\end{equation*}
for all states $\varrho$ and outcomes $x$.
It has been proven in \cite{HeReRyZi18} that $\A^{\va}$ and $\Gamma_p$ are compatible if and only if 
\begin{equation}
|\va | \leq \frac{1}{2} \left( 1- p + \sqrt{(1-p)(1+3p)} \right)  \, .
\end{equation}
This shows that higher is the norm $|\va |$, smaller must $p$ be.
\end{example}

The earlier discussion motivates the following definition, which is central to our investigation. 
\begin{definition}\label{def:comp}
Let $\state_0\subset\state$.
Devices $\D_1,\ldots,\D_n$ are \emph{$\state_0$-compatible} if there exist compatible devices $\D'_1,\ldots,\D'_n$ of the same type such that
 \begin{equation}\label{eq:comp-same}
\D'_j(\varrho) = \D_j(\varrho)
\end{equation}
for all $j=1,\ldots,n$ and states $\varrho\in\state_0$.
Otherwise, $\D_1,\ldots,\D_n$ are \emph{$\state_0$-incompatible}.
\end{definition}

The definition is obviously interesting only when $\D_1,\ldots,\D_n$ are incompatible in the usual sense, i.e., with respect to the full state space.
In that case the definition means that if devices $\D_1,\ldots,\D_n$ are $\state_0$-compatible, their incompatibility cannot be verified by taking test states from $\state_0$ only, and vice versa, if devices $\D_1,\ldots,\D_n$ are $\state_0$-incompatible, their actions on $\state_0$ cannot be simulated by any collection of compatible devices and therefore their incompatibility should be able to be observed in some way.

The $\state_0$-compatibility depends not only on the size of $\state_0$ but also on its structure.
We start with a simple example showing that 
there exist sets $\state_0$ such that 
an arbitrary family of devices is $\state_0$-compatible. 

\begin{example}\label{ex:distinguishable}
Any set of devices $\D_1,\ldots, \D_n$ is $\state_0$-compatible if 
$\state_0=\{\varrho_1, \ldots, \varrho_k\}$ 
consists of perfectly distinguishable states. 
In fact, one may construct a device 
$\D_k'$ which outputs $\D_k(\varrho_j)$ after 
confirming an input state is $\varrho_j$ by measuring 
an observable that distinguishes the states in $\state_0$. 
It is easy to see that the devices $\D_1', \ldots, \D_n'$ are compatible. 
The same argument works for devices in general probabilistic theories and one can use the same reasoning for a subset $\state_0$ that is broadcastable \cite{BaBaLeWi07}.
(We recall that a subset $\state_0$ is broadcastable if there exists a channel $B:\state \to \state\otimes\state$ such that the bipartite state $B(\varrho)$ has marginals equal to $\varrho$ for all $\varrho\in\state_0$.) 
For instance, two qubit states $\id/2$ and $|0\rangle \langle 0|$ are broadcastable even though not distinguishable.
Any pair of qubit channels $\Lambda_1$ and $\Lambda_2$ is $\state_0$-compatible for $\state_0=\{\id/2,|0\rangle \langle 0| \}$ as we can define 
$\Lambda'_j(\varrho) = \sum_{i=0}^1 \langle i | 
\varrho |i\rangle \Lambda_j(|i\rangle\langle i |)$ for $j=1,2$.
The channel $\Lambda'_j$ has clearly the same action as $\Lambda_j$ on $\state_0$.
A joint channel $\Lambda$ for $\Lambda'_1$ and $\Lambda'_2$ is given as
\begin{align*}
\Lambda(\varrho) =  \sum_{i=0}^1 \langle i | 
\varrho |i\rangle \, \Lambda_1(|i\rangle\langle i |) \otimes  \Lambda_2(|i\rangle\langle i |),
\end{align*}
and it is clear that, in fact, $\ptr{2}{\Lambda(\varrho)}=\Lambda_1(\varrho)$ and $\ptr{1}{\Lambda(\varrho)}=\Lambda_2(\varrho)$.
\end{example}

\section{(In)compatibility dimension of devices}\label{sec:dim}

For a subset $\state_0 \subset \state$, we denote by $\bar{\state}_0$ the intersection of the linear hull of $\state_0$ with $\state$, i.e.,
\begin{align*}
\bar{\state}_0 = \{ \varrho \in \state \mid \textrm{$\varrho = \sum_{i=1}^l c_i \varrho_i$ for some $c_i \in \complex$ and $\varrho_i\in\state_0$} \}
\end{align*}
In this definition we can assume without restriction that $c_i\in\real$ and $\sum_i c_i = 1$ as they follow from the positivity and unit-trace of states.
Since the condition \eqref{eq:comp-same} is linear in $\varrho$, we conclude that devices $\D_1,\ldots,\D_n$ are $\state_0$-compatible if and only if they are $\bar{\state}_0$-compatible.
This makes sense: if we can simulate the action of devices for states in $\state_0$, we can simply calculate the action for all states that are linear combinations of those states.
This observation also shows that a reasonable way to quantify the size of a subset $\state_0$ for the task in question is the number of affinely independent states. 

We consider the following questions. 
Given a collection of incompatible devices $\D_1,\ldots,\D_n$,
\begin{itemize}
\item[(a)] what is the smallest subset $\state_0$ such that $\D_1,\ldots,\D_n$ are $\state_0$-incompatible?
\item[(b)] what is the largest subset $\state_0$ such that $\D_1,\ldots,\D_n$ are $\state_0$-compatible?
\end{itemize}
Smallest and largest here mean the number of affinely independent states in $\bar{\state}_0$.
It agrees with the linear dimension of the linear hull of $\state_0$, or $\mathrm{dim}\mathit{aff}\state_{0}+1$, where $\mathrm{dim}\mathit{aff}\state_{0}$ is the affine dimension of the affine hull $\mathit{aff}\state_{0}$ of $\state_0$ \cite{CA97,CO09}.
The answer to (a) quantifies how many states we need to use to detect incompatibility if we choose them carefully, whereas the answer to (b) quantifies how many (affinely independent) states we may have to use if we cannot control their choice. 
Hence for both of these quantities lower number means more incompatibility in the sense of easier detection.
The precise mathematical definitions read as follows.

\begin{definition}
For a collection of incompatible devices $\D_1,\ldots,\D_n$, we denote
\begin{multline*}
	\chi_{incomp}(\D_1,\ldots,\D_n)=
	\min_{\state_{0}\subset\state}
	\{
	\mathrm{dim}\mathit{aff}\state_{0}+1\\
	\mid\mbox{$\D_1,\ldots,\D_n$: $\state_{0}$-incompatible}
	\}
\end{multline*}
and
\begin{multline*}
	\chi_{comp}(\D_1,\ldots,\D_n)=
	\max_{\state_{0}\subset\state}
	\{
	\mathrm{dim}\mathit{aff}\state_{0}+1\\
	\mid\mbox{$\D_1,\ldots,\D_n$: $\state_{0}$-compatible}
	\}.
\end{multline*}
We call these numbers the \emph{incompatibility dimension} and \emph{compatibility dimension} of $\D_1,\ldots,\D_n$, respectively.
\end{definition}

From Example \ref{ex:distinguishable} and the fact that the linear dimension of the linear hull of $\state$ is $d^2$ we conclude that
\begin{equation}\label{eq:incomp-bounds}
2 \leq \chi_{incomp}(\D_1,\ldots,\D_n) \leq d^2 
\end{equation}
and
\begin{equation}\label{eq:comp-bounds}
d \leq \chi_{comp}(\D_1,\ldots,\D_n) \leq d^2-1 \, .
\end{equation}
Further, from the definitions of these quantities it directly follows that
\begin{equation}
\chi_{incomp}(\D_1,\ldots,\D_n) \leq \chi_{comp}(\D_1,\ldots,\D_n) +1 \, .
\end{equation}
We note that based on their definitions, both $\chi_{incomp}$ and $\chi_{comp}$ are expected to be smaller for collections of devices that are more incompatible.

The following monotonicity property of $\chi_{incomp}$ and $\chi_{comp}$ under pre-processing is a basic property that any quantification of incompatibility is expected to satisfy.

\begin{proposition}
Let $\Lambda:\state\to\state$ be a quantum channel and let $\widetilde{\D}_j$ be a pre-processing of $\D_j$ with $\Lambda$ for each $j=1,\ldots,n$, i.e., $\widetilde{\D}_j(\varrho)=\D_j(\Lambda(\varrho)))$.
If $\widetilde{\D}_j$'s are 
incompatible, then also $\D_j$'s are incompatible and 
\begin{equation}
\chi_{incomp}(\widetilde{\D}_1,\ldots,\widetilde{\D}_n) \geq \chi_{incomp}(\D_1,\ldots,\D_n) 
\end{equation}
and
\begin{equation}
\chi_{comp}(\widetilde{\D}_1,\ldots,\widetilde{\D}_n) \geq \chi_{comp}(\D_1,\ldots,\D_n) \, .
\end{equation}
\end{proposition}

\begin{proof}
Suppose that $\D_1, \ldots, \D_n$ are $\state_0$-compatible for some subset $\state_0$. 
Let $\D'$ be a device that gives  devices $\D'_1, \ldots, \D'_n$ as marginals and these marginals satisfy \eqref{eq:comp-same} in $\state_0$.
Then the pre-processing of $\D'$ with $\Lambda$ gives $\widetilde{\D}_1,\ldots,\widetilde{\D}_n$ as marginals in $\state_0$.
The claimed inequalities then follow.
\end{proof}

The post-processing map of a device $\D$ depends on type of the device. 
For instance, the output set of an observable is $P(X)$ and post-processing is then described as a stochastic matrix \cite{MaMu90a}.
We formulate and prove the following monotonicity property of $\chi_{incomp}$ and $\chi_{comp}$ under post-processing only for observables. The formulation is analogous for other types of devices.

\begin{proposition}
	\label{prop:post-processing}
Let $\widetilde{\A}_j$ be a post-processing of $\A_j$ (i.e. $\widetilde{\A}_j(x')=\sum_x \nu_j(x',x)\A_j(x)$ for some stochastic matrix $\nu_j$) for each $j=1,\ldots,n$.
If $\widetilde{\A}_j$'s are $\state_0$-incompatible, then also $\A_j$'s are $\state_0$-incompatible and  
\begin{equation}
\chi_{incomp}(\widetilde{\A}_1,\ldots,\widetilde{\A}_n) \geq \chi_{incomp}(\A_1,\ldots,\A_n) 
\end{equation}
and
\begin{equation}
\chi_{comp}(\widetilde{\A}_1,\ldots,\widetilde{\A}_n) \geq \chi_{comp}(\A_1,\ldots,\A_n) \, .
\end{equation}
\end{proposition}

\begin{proof}
Suppose that $\A_1, \ldots, \A_n$ are 
$\state_0$-compatible for some subset $\state_0$. 
This means that there exists an observable 
$\G$ satisfying for all $\varrho\in\state_0$, any $j$ and $x_j$, 
\begin{equation}\label{eq:pp-1}
\tr{\varrho \A_j(x_j)}
= \sum_{l\neq j} \sum_{x_l} 
\tr{\varrho \G(x_1, \ldots, x_n)} \, .
\end{equation}
We define an observable 
$\widetilde{\G}$ as 
$\widetilde{\G}(x'_1, \ldots, x'_n)
= \sum_{x_1, \ldots, x_n}
\nu(x'_1|x_1)\cdots \nu(x'_n|x_n)
\G(x_1, \ldots, x_n)$, 
and it then satisfies 
\begin{eqnarray}
\tr{\varrho \widetilde{\A}_j(x'_j)}
= \sum_{l \neq j} \sum_{x'_l}
\tr{\varrho \widetilde{\G}(x'_1, \ldots, x'_n)}
\end{eqnarray}
for all $\varrho\in\state_0$, any $j$ and $x'_j$.
This shows that $\widetilde{\A}_1,\ldots,\widetilde{\A}_n$ are $\state_0$-compatible.
The claimed inequalities then follow.
\end{proof}

We will now have some examples to demonstrate the values of $\chi_{incomp}$ and $\chi_{comp}$ in some standard cases.

\begin{example}\label{ex:identity}
Let us consider the identity channel $\mbox{id}: \mathcal{S}(\mathbb{C}^d)
\to \mathcal{S}(\mathbb{C}^d)$. It follows from the definitions that two identity channels are $\state_0$-compatible if and only if $\state_0$ is a broadcastable set. 
It is known that a subset of states is broadcastable only if the states commute with each other \cite{Barnumetal96}, and for this reason the pair of two identity channels is $\state_0$-incompatible whenever $\state_0$ contains two noncommuting states. 
Therefore, we have $\chi_{incomp}(\mbox{id}, \mbox{id}) =2$.
On the other hand, $\state_0$ consisting of 
distinguishable states makes the identity channels
$\state_0$-compatible. 
As $\state_0$ consisting of commutative states 
has at most $d$ affinely independent states, 
we conclude that $\chi_{comp}(\mbox{id}, \mbox{id})=d$. 
\end{example}

A comparison of the results of Example \ref{ex:identity} to the bounds \eqref{eq:incomp-bounds} and \eqref{eq:comp-bounds} shows that the pair of identity channels has the smallest possible incompatibility and compatibility dimensions. 
This is quite expectable as that pair is consider to be the most incompatible pair - any device can be post-processed from the identity channel. 
Perhaps surprisingly, the lower bound of $\chi_{incomp}$ can be attained already with a pair of dichotomic observables; this is shown in the next example.
 
\begin{example}\label{ex:pq}
Let $P$ and $Q$ be two noncommuting one-dimensional projections in a $d$-dimensional Hilbert space $\hi$.
We define two dichotomic observables $\A$ and $\B$
as 
\begin{align*}
\A(1)=P \, , \A(0)=\id-P \, , \quad \B(1)=Q \, , \B(0)=\id-Q \, .
\end{align*} 
Let us then consider a subset consisting of two states, 
$$
\state_0=\{\varrho^P, \varrho^Q\}:=
\{\tfrac{1}{d-1}(\id-P), \tfrac{1}{d-1}(\id-Q)\} \, .
$$
We find that the dichotomic observables $\A$ and $\B$ are $\state_0$-incompatible. 
To see this, let us make a counter assumption that  $\A$ and $\B$ are $\state_0$-compatible, in which case there exists $\G$ such that the marginal condition \eqref{eq:naive} holds for both observables and for all $\varrho\in\state_0$.
We have $\tr{\varrho^P \A(1)}=0$ and therefore
\begin{align*}
0=\tr{(\id-P) \G(1,1)} = \tr{(\id-P) \G(1,0)}.
\end{align*}
It follows that $\G(1,1)=\alpha P$ and $\G(1,0) =\beta P$.
Further, $\tr{P \A(1)} =1$ and hence $\alpha+\beta=1$.
In a similar way we obtain $\G(1,1)=\gamma Q$ and $\G(0,1) =\delta Q$ with $\gamma+\delta=1$.
It follows that $\alpha=\gamma=0$ and $\beta=\delta=1$.
But $\G(1,0) + \G(0,1) = P + Q$ contradicts $\G(1,0) + \G(0,1) \leq \id$.
Thus we conclude $\chi_{incomp}(\A, \B)
=2$. 
\end{example}
 
For two incompatible sharp qubit observables (Example \ref{ex:qubit-1}) the previous example gives a concrete subset of two states such that the observables are incompatible and proves that $\chi_{incomp}(\A^{\va},\A^{\vb})=2$ for such a pair.
The incompatibility dimension for unsharp qubit observables is more complicated and will be treated in Sec. \ref{sec:qubit}.

\begin{example}\label{ex:fix}
Let us consider two observables $\A$ and $\B$.
Fix a state $\varrho_0\in\state$ and define 
\begin{equation*}
\state_0 = \{ \varrho\in\state : \tr{\varrho\A(x)} = \tr{\varrho_0\A(x)} \  \forall x \} \, .
\end{equation*}
Then $\A$ and $\B$ are $\state_0$-compatible. 
To see this, we define an observable $\G$ as
\begin{equation*}
\G(x,y) =  \tr{\varrho_0\A(x)} \B(y) \, .
\end{equation*}
It is then straightforward to verify that \eqref{eq:naive} holds for all $\varrho \in \state_0$.

As a special instance of this construction, let $\A^{\va}$ be a qubit observable and $\va\neq 0$ (see Example \ref{ex:qubit-1}).
We choose $\state_0=\{ \varrho\in\state\mid \tr{\varrho\A^{\va}(+)}=\half \}$.
We then have $\state_0 = \{ \half(\id + \vr\cdot\vsigma)\mid \vr\cdot\va=0 \}$ and hence
$\mathrm{dim}\mathit{aff}\state_0=2$.
Based on the previous argument, $\A^{\va}$ is $\state_0$-compatible with any $\A^{\vb}$.
Therefore, $\chi_{comp}(\A^{\va},\A^{\vb})=3$ for all incompatible qubit observables $\A^{\va}$ and $\A^{\vb}$.
\end{example}

\subsection*{Remark on other formulations of incompatibility dimension}

The notion of $\state_0$-compatibility for quantum observables has  been introduced in \cite{GuQuAo19} and in that particular case (i.e. quantum observables) it is equivalent to Def. \ref{def:comp}. 
In the current investigation our focus is on the largest or smallest $\state_0$ on which devices $\D_1,\ldots,\D_n$ are compatible or incompatible, and this has some differences to the earlier approaches. In \cite{LoNe21}, the term ``compatibility dimension" was introduced and
for observables $\A_1,\ldots,\A_n$ on a $d$ dimensional Hilbert space $\hi=\complex^{d}$ it is
	\begin{align*}
		R(\A_1,\ldots,\A_n)=\max\{r\le d\mid
		\exists V\colon\complex^{r}\to\complex^{d}\ isometry\ \\
		s.t.\  V^{*}\A_1V, ,\ldots,V^{*}\A_n V\  are\  compatible
		\},
	\end{align*}
 Evaluations of $R(\A_1,\ldots,\A_n)$ in various cases such as $n=2$ and $\A_{1}$ and $\A_{2}$ are rank-1 were presented in \cite{LoNe21}.
	To describe it in our notions, let us denote $\complex^{r}$ by $\hik$, and define $\state_{\hi}$ and $\state_{\hik}$ as the set of all density operator on $\hi$ and $\hik$ respectively. 
	We also introduce $\state_{V\hik}$ as
	\begin{align*}
		\state_{V\hik}:=\{\varrho\in\state\mid \mathrm{supp}\varrho\subset V\hik\}=
		V\state_{\hik}V^{*}\subset\state_{\hi}.
	\end{align*}
	Then, we can see that the $\state_{\hik}$-compatibility of $V^{*}\A_1V, ,\ldots,V^{*}\A_n V$ is equivalent to the $\state_{V\hik}$-compatibility of $\A_1,\ldots,\A_n$.
	Therefore, if we focus only on sets of states such as $\state_{V\hik}$ (i.e. states with fixed support), then there is no essential difference between our compatibility dimension and the previous one: $R(\A_1,\ldots,\A_n)=r$ iff $\chi_{comp}(\A_1,\ldots,\A_n)=r^{2}$. 
	In \cite{LoNe21} also the concept of ``strong compatibility dimension" was defined as
	\begin{align*}
		\overline{R}(\A_1,\ldots,\A_n)=\max\{r\le d\mid
		\forall V\colon\complex^{r}\to\complex^{d}\ isometry\ \\
		s.t.\  V^{*}\A_1V, ,\ldots,V^{*}\A_n V\  are\  compatible
		\}.
	\end{align*}
	It is related to our notion of incompatibility dimension.
	In fact, if we only admit sets of states such as $\state_{V\hik}$, then $\overline{R}(\A_1,\ldots,\A_n)$ and $\chi_{incomp}(\A_1,\ldots,\A_n)$ are essentially the same: $\overline{R}(\A_1,\ldots,\A_n)=r$ iff $\chi_{incomp}(\A_1,\ldots,\A_n)=(r+1)^{2}$.

Similar notions have been introduced and investigated also in \cite{Kiukas20,UoKrDeMiTaPeGuBr21}.
	As in \cite{LoNe21}, these works focus on quantum observables and on subsets of states that are lower dimensional subspaces of the original state space.
	Therefore, the notions are not directly applicable in GPTs.
In \cite{UoKrDeMiTaPeGuBr21} incompatibility is classified into three types.
They are explained exactly in terms of \cite{LoNe21} notion as\\
(i) incompressive incompatibility: $(\A_1,\ldots,\A_n)$ are $\state_{V\hik}$-compatible for all $\hik$ and $V$\\
(ii) fully compressive incompatibility: $(\A_1,\ldots,\A_n)$ are $\state_{V\hik}$-incompatible for all nontrivial $\hik$ and $V$\\
(iii) partly compressive incompatibility: there is a $V$ and $\hik$ such that $(\A_1,\ldots,\A_n)$ are $\state_{V\hik}$-compatible, and some $V'$ and $\hik'$ such that $(\A_1,\ldots,\A_n)$ are $\state_{V'\hik'}$-incompatible.\\
In \cite{UoKrDeMiTaPeGuBr21} concrete constructions of these three types of incompatible observables were given.

\section{Relation between incompatibility dimension and incompatibility witness for observables}

In this section we show how the notion of incompatibility dimension is related to the notion of incompatibility witness.
An \emph{incompatibility witness} is an affine functional $\xi$ defined on $n$-tuples of observables such that $\xi$ takes non-negative values on all compatible $n$-tuples and a negative value at least for some incompatible $n$-tuple \cite{Jencova18,CaHeTo19,CaHeMiTo19JMP}.
Every incompatibility witness $\xi$ is of the form 
\begin{align}
	\label{eq:standard-form0}
	\xi(\oplus_{j=1}^{n} \A_j)=\delta - 
	f(\oplus_{j=1}^{n} \A_j),
\end{align}
where $\delta\in \mathbb{R}$ and $f$ is a linear functional on $\oplus_{j=1}^n \mathcal{L}_s(\hi)^{m_j}$ with $\mathcal{L}_s(\hi)$ being the set of all 
self-adjoint operators on $\hi$ and $m_j$ the number of outcomes of $\A_j$.
It can be written also in the form
\begin{equation}
	\label{eq:standard-form}
	\xi(\A_1,\ldots,\A_n) = \delta - \sum_{j=1}^n \sum_{x_j=1}^{m_j} c_{j,x_j} \tr{\varrho_{j,x_j}\A_j(x_j)},
\end{equation}
where $c_{j,x_j}$'s are real numbers, and $\varrho_{j,x_j}$'s are states.
This result has been proven in \cite{CaHeTo19} for incompatibility witnesses acting on pairs of observables and the generalization to $n$-tuples is straightforward.
A witness $\xi$ \emph{detects} the incompatibility of observables $\A_1,\ldots,\A_n$ if $\xi(\A_1,\ldots,\A_n) <0$.
The following proposition gives a simple relation between incompatibility dimension and incompatibility witness.
\begin{proposition}
\label{prop:dim-witness}
Assume that an incompatibility witness $\xi$ has the form \eqref{eq:standard-form} and it detects the incompatibility of observables $\A_1,\ldots,\A_n$.
Then $\A_1,\ldots,\A_n$ are $\state_0$-incompatible for $\state_0 = \{ \varrho_{j,x_j}\mid j=1,\ldots,n, x_j=1,\ldots,m_j\}$.
\end{proposition}
\begin{proof}
Let $\A_1,\ldots,\A_n$ be $\state_0$-compatible.
Then we would have compatible observables $\widetilde{\A}_1,\ldots,\widetilde{\A}_n$ such that $\tr{\varrho\A_j(x_j)}=\tr{\varrho\widetilde{\A}_j(x_j)}$ for all $\varrho\in\state_0$.
This would imply that 
\begin{equation*}
	\xi(\A_1,\ldots,\A_n) = \xi(\widetilde{\A}_1,\ldots,\widetilde{\A}_n) \geq 0 \, , 
\end{equation*}
which contradicts the assumption that $\xi$ detects the incompatibility of observables $\A_1,\ldots,\A_n$.
\end{proof}
It has been shown in \cite{CaHeTo19} that any incompatible pair of observables is detected by some incompatibility witness of the form \eqref{eq:standard-form}.
The proof is straightforward to generalize to $n$-tuples of observables, and thus, together with Proposition \ref{prop:dim-witness}, we can obtain
\begin{equation}
	\label{eq:upper bound}
	\chi_{incomp}(\A_1,\ldots,\A_n) \leq m_1 + \cdots + m_n.
\end{equation}
That is, the incompatibility dimension of $\A_1,\ldots,\A_n$ can be evaluated via their incompatibility witness (we will derive a better upper bound later in this section).
We can further prove the following proposition.

\begin{proposition}
	\label{prop:dim-witness0}
	The statements (i) and (ii) for a set of incompatible observables $\{\A_1, \ldots, \A_n\}$ are equivalent:
	\begin{itemize}
		\item[(i)]$\chi_{incomp}(\A_1, \ldots, \A_n)\leq N$ 
		\item[(ii)]There exist a family of linearly independent 
		states $\{\varrho_1, \ldots, \varrho_N\}$ and real numbers $\delta$ and $\{c_{l,j,x_j}\}_{l, j, x_{j}}$ 
		$(l=1,\ldots, N, j=1,\ldots,n, x_j=1,\ldots,m_j)$ such that the incompatibility witness $\xi$ defined by
		\begin{align*}
			\xi(\B_1, \ldots, \B_n)
			=\delta- \sum_{l=1}^{N}\sum_{j=1}^{n} \sum_{x_j=1}^{m_j} 
			c_{l,j,x_j}\mbox{tr}[\varrho_l \B_j(x_j)]
		\end{align*}
	 detects the incompatibility of $\{\A_1, \ldots, \A_n\}$.
	\end{itemize}
\end{proposition}

The claim $\mathit{(i)}\Rightarrow\mathit{(ii)}$ may be regarded as the converse of the previous argument to obtain \eqref{eq:upper bound}.
It manifests that we can find an incompatibility witness detecting the incompatibility of $\{\A_1, \ldots, \A_n\}$ reflecting their incompatibility dimension.
\begin{proof}
$\mathit{(ii)}\Rightarrow\mathit{(i)}$ can be proved in the same way as Proposition \ref{prop:dim-witness}.
Thus we focus on proving $\mathit{(i)}\Rightarrow\mathit{(ii)}$.

Suppose that a family of observable $\{\A_1, \ldots, \A_n\}$ satisfies $\chi_{incomp}(\A_1, \ldots, \A_n)= N$.
Then there exists a family of linearly independent states $\{\varrho_1, \varrho_2, \ldots, \varrho_N\}$ in $\mathcal{L}_s(\hi)$ on which $\{\A_1, \ldots, \A_n\}$ are incompatible.
We can regard the family $\{\A_1, \ldots, \A_n\}$ as an element of a vector space $\mathcal{L}$ defined as $\mathcal{L}:=\oplus_{j=1}^n \mathcal{L}_s(\hi)^{m_j}$, that is, $\A:=\oplus_{j=1}^n\A_j\in \mathcal{L}$.
For each $l=1,\ldots,N$, $j=1, \ldots, n$, and $x_j=1, \ldots, m_j$, let us define a subset $K(\A, \varrho_l, j,x_j)$ of 
$\mathcal{L}$ as  
\begin{align}
		\label{eq:def of K}
\begin{aligned}
	&K(\A, \varrho_l, j,x_j)\\
	&\qquad:=\{\B\in\mathcal{L}\mid\langle \varrho_l | \B_j(x_j)\rangle_{HS}
	= \langle \varrho_l | \A_j(x_j)\rangle_{HS}\}, 
\end{aligned} 
\end{align}
where $\langle \varrho_l | \A_j(x_j)\rangle_{HS}
:=\mbox{tr}[\varrho_l\A_j(x_j)]$ is the Hilbert-Schmidt inner 
product on $\mathcal{L}_s(\hi)$.
Note that this inner product can be naturally extended to an inner product $\langle\langle\cdot | \cdot\rangle\rangle$ on $\mathcal{L}$: 
\begin{align*}
\langle\langle\A | \B\rangle\rangle=\sum_{j=1}^{n}\sum_{x_{j}=1}^{m_j}\langle \A_j(x_j)|\B_j(x_j)\rangle_{HS} \, .
\end{align*}
Embedding $\varrho_{l}$ into $\mathcal{L}$ by $\hat{\varrho}_l^{j,x_j}=
\oplus_{i=1}^{n} \oplus_{y=1}^{m_{i}} \delta_{ij} \delta_{yx_j} 
\varrho_l$ for each $j,x_j$ and $l$, we obtain another representation 
of \eqref{eq:def of K} as 
\begin{align}
	\label{eq:def of K2}
	K(\A, \varrho_l, j,x_j)
	=\{\B\mid\langle \langle \hat{\varrho}_l^{j,x}| \B\rangle\rangle
	= \langle \langle \hat{\varrho}_l^{j,x_j}|\A\rangle\rangle\} \, .
\end{align}
Thus this set is a hyperplane in $\mathcal{L}$. 
Note that $\{\hat{\varrho}_l^{j,x}\}_{l,j,x_j}$ is a linearly independent set in $\mathcal{L}$.
Consider an affine set 
$K:=\cap_{l=1}^N \cap_{j=1}^{n}\cap_{x_{j}=1}^{m_j} K(\A,\varrho_l,j,x_j)$. 
Because $\{\A_1, \ldots, \A_n\}$ is incompatible in $\{\varrho_{1},\cdots,\varrho_N\}$, it satisfies 
\begin{align}
	\label{eq:KC-empty} 
	K\cap C=\emptyset,
\end{align}
where $C:=\{\C\in\mathcal{L}\mid\mbox{$\{\C_{1},\ldots\C_n\}$ is compatible} \}$.
Thus, by virtue of the separating hyperplane theorem \cite{CA97}, there exists a hyperplane in $\mathcal{L}$ which separates strongly the (closed) convex sets $K$ and $C$.
In the following, we will show that one of those separating hyperplanes can be constructed from $\{\hat{\varrho}_l^{j,x}\}_{l,j,x_j}$.

Let us extend a family of linearly 
independent vectors 
$\{\hat{\varrho}_l^{j,x_j}\}_{l, j, x_j}$ to form a basis of $\mathcal{L}$.  
That is, we introduce a 
basis 
$\{v_b\}_{b=1,\ldots, 
	\dim \mathcal{L}}$  of $\mathcal{L}$ satisfying  
$\{v_a\}_{a=1,\ldots, 
	N (\sum_j m_j)}= \{\hat{\varrho}_l^{j,x_j}\}_{l, j, x_j}$. 
We introduce its dual basis $\{w_b\}_{
	b=1,2,\ldots, \dim \mathcal{L}}$
satisfying
$\langle \langle v_a|
w_b\rangle \rangle=\delta_{ab}$. 
Because $K$ can be written as
\begin{align*}
K=\{\B\mid
\langle \langle
\hat{\varrho}_l^{j,x_j} |(\B-\A)\rangle\rangle
=0, \forall l, j,x_j \},
\end{align*}
it is represented in terms this (dual) basis as 
\begin{align*}
K=\A+K_0, 
\end{align*}
where $K_0$ is an affine set defined by 
\begin{align}
	K_0:&=
	\{ \sum_{a=N(\sum_j m_j)+1}^{\dim\mathcal{L}}c_a w_a \mid c_a\in\real\}
	\label{eq:expresison_of_K_0}
\end{align}
Now we can construct a hyperplane separating $K$ and $C$.
To do this, let us focus on the convex sets $K_0$ and $C':=C-\A$ instead of $K$ and $C$, which satisfy
$K_0\cap C'=\emptyset$ because of \eqref{eq:KC-empty}.
We can apply the separating hyperplane theorem (Theorem 11.2 in \cite{CA97} for the affine set $K_0$ and convex set $C'$.
There exists a hyperplane $H_0$ in $\mathcal{L}$ such that $K_0$ and $C'$ are contained by $H_0$ and one of its associating open half-spaces respectively.
That is, there exists $h\in\mathcal{L}$ satisfying
\begin{align*}
	\label{eq:hyperplane2}
	H_0=\{\B\in \mathcal{L}\mid\langle \langle
	\B |h\rangle\rangle=0
	\}
\end{align*}
with $K_0\subset H_0$, and $\langle \langle
\C' |h\rangle\rangle<0$ for all $\C'\in C'$.
Let us examine the vector $h$.
It satisfies 
\[
\langle \langle
w_{a} |h\rangle\rangle=0\ \ \mbox{for all $a=N(\sum_j m_j)+1,\ldots, \dim\mathcal{L}$}
\]
because $K_0\subset H_0$ (see \eqref{eq:expresison_of_K_0}).
Thus, if we write $h$ as $h=\sum_{a=1}^{\dim\mathcal{L}}c_a v_{a}$, then we can find that $c_a=0$ holds for all $a=N(\sum_j m_j)+1,\ldots, \dim\mathcal{L}$.
It follows that
\begin{align*}
	h=\sum_{a=1}^{N(\sum_j m_j)}c_a v_a=\sum_l\sum_j\sum_{x_j}c_{l, j, x_j}\hat{\varrho}_l^{j,x_j}
\end{align*}
holds, and the hyperplane $H_0$ can be written as
\begin{align*}
	\label{eq:hyperplane}
	H_0=\{\B\in \mathcal{L}\mid\sum_l\sum_j\sum_{x_j}c_{l, j, x_j}\tr{\varrho_l \B_j(x_j)}=0
	\}.
\end{align*}
Then, the hyperplane $H':=\A+H_0$, a translation of $H_0$, of the form
\[
H'=\{\B\in \mathcal{L}\mid\sum_l\sum_j\sum_{x_j}c_{l, j, x_j}\tr{\varrho_l \B_j(x_j)}=\delta'
\}
\]
contains the original sets $K$, and satisfy
\[
\sum_l\sum_j\sum_{x_j}c_{l, j, x_j}\tr{\varrho_l \C_j(x_j)}<\delta'
\] 
for all $\C\in C$.
We can displace $H'$ slightly in the direction of $C$ to obtain a hyperplane $H$ defined as
\[
H=\{\B\in \mathcal{L}\mid\sum_l\sum_j\sum_{x_j}c_{l, j, x_j}\tr{\varrho_l \B_j(x_j)}=\delta
\},
\]
which (strongly) separates $H'$ (in particular $K$) and $C$ because $H'$ is closed and $C$ is compact (see Corollary 11.4.2 in \cite{CA97}).
The claim now follows as $\A\in K$.
\end{proof}

\subsection*{An upper bound on the incompatibility dimension of observables via incompatibility witness}

We can give a better upper bound than \eqref{eq:upper bound} for the incompatibiliy dimension by slightly modifing the previous argument in \cite{CaHeTo19} on incompatibility witness.
\begin{proposition}
	\label{prop:upper-bound}
	Let $\A_1,\ldots,\A_n$ be incompatible observables with $m_1,\ldots,m_n$ outcomes, respectively. 
	Then
	\begin{align*}
		\chi_{incomp}(\A_1,\ldots,\A_n) \leq 
		\sum_{j=1}^n m_j -n+1.
	\end{align*}
\end{proposition}
\begin{proof}
We continue following the same notations as the proof of Proposition \ref{prop:dim-witness0}.
Let us assume the incompatibility of $\A_1,\ldots,\A_n$ are detected by an incompatibility witness $\xi$.
The functional $\xi$ is of the form 
\begin{align*}
		\xi(\A)=\delta - 
		f(\A)
\end{align*}
with a real number $\delta$ and a functional $f$ on $\mathcal{L}$ (see \eqref{eq:standard-form0}).
Then, Riesz representation theorem shows that 
the functional $f$ can be represented as 
\begin{align*}
f(\A) 
&=\sum_{j=1}^n \sum_{x_j}^{m_j}
\langle F_j(x_j) | \A_j(x_j)\rangle_{HS}
\end{align*}
with some $F_j(x_j) \in \mathcal{L}_s(\hi)$ $(j=1, \ldots, n,\ x_j=1, \ldots, m_{j})$.
If we define $F'_j(x_j) = F_j(x_j) +\epsilon_j \id$, then we find 
\begin{align*}
	\xi(\A)=\delta + d \sum_j \epsilon_j 
	- \sum_{j=1}^n \sum_{x_j=1}^{m_j} 
	\langle F'_j(x_j) | \A_j(x_j)\rangle_{HS}.  
\end{align*}
We choose $\epsilon_j$ so that 
$$
\sum_{x_j}\mbox{tr}[F'_j(x_j)] =\sum_{x_j} \langle F'_j (x_j)| \id\rangle_{HS}=0
$$
holds.  
The choice of $\{F'_j(x_j)\}_{j, x_j}$ has still some 
freedom. 
Each $F'_j(x_j)$ can be replaced with 
$F''_j(x_j)=F'_j(x_j) + T_j$, where $T_j
\in \mathcal{L}_s(\hi)$ satisfies $\mbox{tr}[T_j]
=\langle T_j |\id\rangle_{HS}=0$. 
In fact, it holds that
\begin{eqnarray*}
	&&
	\sum_{x_j} \langle F''_j(x_j)|\A_j(x_j)\rangle_{HS} 
	\\
	&&
	=
	\sum_{x_j} \langle F'_j (x_j) |\A_j(x_j)\rangle_{HS} 
	+ \sum_{x_j} \langle T_j|\A_j(x_j)\rangle_{HS} 
	\\
	&&
	= \sum_{x_j} \langle F'_j(x_j) | \A_j(x_j)\rangle_{HS} 
	+ \langle T_j |\id\rangle_{HS} 
	\\
	&&
	=\sum_{x_j} \langle F'_j(x_j) | \A_j(x_j)\rangle_{HS}.  
\end{eqnarray*}
We choose $T_j$ as $m_j T_j = - \sum_{x_j=1}^{m_j} F'_j(x_j)$ which indeed satisfies 
$m_j \langle T_j | \id\rangle_{HS}= -\sum_{x_j=1}^{m_j} \langle F'_j(x_j)|\id\rangle_{HS}=0$, 
i.e. $\tr{T_{j}}=0$,
to obtain 
\begin{eqnarray*}
	\sum_{x_j} F''_j(x_j) =0. 
\end{eqnarray*}
We further choose 
large numbers $\alpha_j \geq 0$ 
so that $G_j(x_j):=F''_j(x_j) + \alpha_j 
\id \geq 0$ for all $j$ and $x_j$.
Now we obtain a representation of 
the witness which is equivalent to $\xi$ for $n$-tuples of observables
as 
\begin{multline*}
	\xi^*(\A)
	= \delta + d \sum_j (\epsilon_j +\alpha_j)  
	\\
	- \sum_j \sum_{x_j} \langle G_j(x_j) |\A_j(x_j)\rangle_{HS}, 
\end{multline*} 
where positive operators 
$G_j(x_j)$'s satisfy $\sum_{x_j} 
G_j (x_j) = m_j \alpha_j \id$. 
Defining density operators 
$\varrho_j(x_j)$ by 
$\varrho_j(x_j) = \frac{G_j(x_j)}{\mbox{tr}[G_j(x_j)]}$, 
we obtain yet another representation 
\begin{multline*}
	\xi^*(\A) 
	= 
	\delta + d\sum_j (\epsilon_j + \alpha_j)
	\\
	- \sum_j \sum_{x_j} 
	\mbox{tr}[G_j(x_j)] \mbox{tr}[\varrho_j(x_j)
	\A_j(x_j)]
\end{multline*} 
with $\varrho_j(x_j)$'s satisfying 
constraints  
\begin{eqnarray}
	\label{eq:witness constraint}
	\sum_{x_j} \mbox{tr}[G_j(x_j)] \varrho_j(x_j)
	= m_j \alpha_j \id. 
\end{eqnarray}
Thus, according to Proposition \ref{prop:dim-witness}, $\A_1,\ldots,\A_n$ are $\state_{0}$-incompatible with $\state_{0}=\{\varrho_j(x_j)\}_{j, x_{j}}$.
To evaluate $\mathrm{dim}\mathit{aff}\state_{0}$, we focus on the condition \eqref{eq:witness constraint}.
Introducing parameters $p_j(x_{j}):=\mbox{tr}[G_j(x_j)] /dm_j \alpha_j$ such that $\sum_{x_j}p_j(x_{j})=1$, we obtain
\[
\sum_{x_j} p_j(x_{j}) \varrho_j(x_j)
= \frac{1}{d}\id,
\]
or
\[
\sum_{x_j} p_j(x_{j}) \tilde{\varrho_j}(x_j)=0,
\]
where $\tilde{\varrho_j}(x_j):=\varrho_j(x_j)-\frac{1}{d}\id$.
It follows that $\{\tilde{\varrho_j}(x_j)\}_{x_{j}}$  are linearly dependent, and thus
\[
\mathrm{dim}\mathit{span}\{\tilde{\varrho_j}(x_j)\}_{x_{j}}\le m_{1}-1.
\]
Similar arguments for the other $j$'s result in
\[
\mathrm{dim}\mathit{span}\{\tilde{\varrho_j}(x_j)\}_{j, x_{j}}\le \sum_{j}(m_{j}-1)=\sum_{j}m_{j}-n.
\]
Considering that 
\[
\mathrm{dim}\mathit{span}\{\tilde{\varrho_j}(x_j)\}_{j, x_{j}} =\mathrm{dim}\mathit{aff}\{\varrho_j(x_j)\}_{j, x_{j}}
\]
holds, we can obtain the claim of the proposition. 
\end{proof}

The bound in Proposition \ref{prop:upper-bound} is not tight in general since the right-hand side of the inequality can exceed the bound obtained in \eqref{eq:incomp-bounds}.  
However, for small $n$ and $m_j$'s, the bound can be tight. 
In fact, while for $n=2$ and $m_1=m_2=2$ it gives $\chi_{incomp}(\A_1, \A_2)\leq 3$, we will construct an example which attains this upper bound in the next section.  

\section{Mutually unbiased qubit observables}\label{sec:qubit}

In this section we study the incompatibility dimension of pairs of unbiased qubit observables introduced in Example \ref{ex:qubit-1}.
We concentrate on pairs that are mutually unbiased, i.e., $\tr{\A^\va(\pm) \A^\vb(\pm)} = 1/2$.
(This terminology originates from the fact that if the observables are sharp, then the respective orthonormal bases are mutually unbiased. In the previously written form the definition makes sense also for unsharp observables \cite{BeBuBuCaHeTo13}.)
The condition of mutual unbiasedness is invariant under a global unitary transformation, hence it is enough to fix the basis $\vx=(1,0,0)$, $\vy=(0,1,0)$, $\vz=(0,0,1)$ in $\real^3$ and choose two of these unit vectors. 
We will study the observables $\A^{t\vx}$ and $\A^{t\vy}$, where $0\leq t \leq 1$.
The observables are written explicitly as
\begin{align*}
	\A^{t\vx}(\pm)=
	\frac{1}{2}(\id \pm t \sigma_1) \, , \quad 
	\A^{t\vy}(\pm)=
	\frac{1}{2}(\id \pm t \sigma_2).
\end{align*}
The condition \eqref{eq:paul} shows that $\A^{t\vx}$ and $\A^{t\vy}$ 
are incompatible if and only if $1/\sqrt{2} < t \leq 1$. 
The choice of having mutually unbiased observables as well as using a single noise parameter instead of two is to simplify the calculations. 

We have seen in Example \ref{ex:fix} that $\chi_{comp}(\A^{t\vx},\A^{t\vy})=3$ for all values $t$ for which the pair is incompatible.
We have further seen (discussion after Example \ref{ex:pq}) that $\chi_{incomp}(\A^{\vx},\A^{\vy})=2$, 
and from Prop. \ref{prop:upper-bound} follows that $\chi_{incomp}(\A^{t\vx},\A^{t\vy}) \leq 3$ for all $1/\sqrt{2} < t \leq 1$.
The remaining question is then about the exact value of $\chi_{incomp}(\A^{t\vx},\A^{t\vy})$, which can depend on the noise parameter $t$ and will be in our focus in this section (see Table \ref{table_incomp}). 

\begin{table}[h]
	\centering
	\begin{tabular}{c||c|c}
		& $\chi_{incomp}(\A^{t\vx},\A^{t\vy})$ & $\chi_{comp}(\A^{t\vx},\A^{t\vy})$ \\ \hline
		$t\le\frac{1}{\sqrt{2}}$ & - & - \\ \hline
		$\frac{1}{\sqrt{2}}<t<1$ 
		& \begin{tabular}{c}
			2 or 3\\
			(Proposition \ref{prop:qubit-threshold})
		\end{tabular}	 
		& \multirow{2}{*}{
			\begin{tabular}{c}
				3\\
				(Example \ref{ex:fix})
			\end{tabular}	
		} \\ \cline{1-2}
		$t=1$ 
		& \begin{tabular}{c}
			2\\
			(Example \ref{ex:pq})
		\end{tabular}	 
		&  \\ \hline
	\end{tabular}
	\caption{$\chi_{incomp}$ and $\chi_{comp}$ for $(\A^{t\vx},\A^{t\vy})$ with $0\le t\le1$. For $t\le1/\sqrt{2}$ the observables $\A^{t\vx}$ and $\A^{t\vy}$ are compatible and $\chi_{incomp}$ and $\chi_{comp}$ are not defined.}
	\label{table_incomp}
\end{table}

Let us first make a simple observation that follows from Prop. \ref{prop:post-processing}. 
Considering that $\A^{s\vx}$ is obtained as a post-processing of $\A^{t\vx}$ if and only if $s\leq t$, we conclude that
\begin{align*}
	\chi_{incomp}(\A^{s\vx},\A^{s\vy})=2\qquad\qquad\qquad\qquad\qquad\qquad\\
	\Rightarrow\chi_{incomp}(\A^{t\vx},\A^{t\vy})=2 \quad \textrm{for $\frac{1}{\sqrt{2}}<s \leq t$} \, , 
\end{align*}
and 
\begin{align*}
	\chi_{incomp}(\A^{s' \vx},\A^{s' \vy})=3\qquad\qquad\qquad\qquad\qquad\qquad\\
	\Rightarrow\chi_{incomp}(\A^{t'\vx},\A^{t'\vy})=3 \quad \textrm{for $s'\geq t'>\frac{1}{\sqrt{2}}$} \, .
\end{align*}
Interestingly, there is a threshold value $t_0$ where the value of $\chi_{incomp}(\A^{t\vx},\A^{t\vy})$ changes; this is the content of the following proposition.

\begin{proposition}\label{prop:qubit-threshold}
	There exists $1/\sqrt{2}<t_0<1$ such that $\chi_{incomp}(\A^{t\vx},\A^{t\vy})=3$ for $1/\sqrt{2}<t\le t_0$ and $\chi_{incomp}(\A^{t\vx},\A^{t\vy})=2$ for $t_0<t \leq 1$.
\end{proposition}

The main line of the lengthy proof of Prop. \ref{prop:qubit-threshold} is the following. 
Defining two subsets $L$ and $M$ of $(\frac{1}{\sqrt{2}}, 1]$ as
\begin{align}
	\label{eq: sets of xi=2,3}
	\begin{aligned}
		L&:=\{t\mid\chi_{incomp}(\A^{t\vx},\A^{t\vy})=2\},\\
		M&:=\{t\mid\chi_{incomp}(\A^{t\vx},\A^{t\vy})=3\},
	\end{aligned}
\end{align}
we see that
\begin{align}
	\inf L=\sup M(=:t_{0}')
\end{align}
holds unless $L$ and $M$ are empty.
By its definition, the number $t_{0}'$ satisfies
\begin{align*}
	&\chi_{incomp}(\A^{t\vx},\A^{t\vy})=2\ \ \mbox{for\  $t>t_{0}'$},\\
	&\chi_{incomp}(\A^{t\vx},\A^{t\vy})=3\ \ \mbox{for\  $t<t_{0}'$}.
\end{align*}
Based on the considerations above, the proof of Proposition \ref{prop:qubit-threshold} proceeds as follows.
First, in Part 1, we prove that $M$ is nonempty while $L$ has already been shown to be nonempty as $t=1\in L$.
It will be found that $\chi_{incomp}(\A^{t\vx},\A^{t\vy})=3$ for $t$ sufficiently close to $\frac{1}{\sqrt{2}}$, and thus $t_{0}'$ introduced above can be defined successfully.
Then, we demonstrate in Part 2 that $\sup M=\max M$, i.e. $t_{0}'$ equals to $t_{0}$ in the claim of Prop. \ref{prop:qubit-threshold}.

\begin{remark}
	In \cite{GuQuAo19} a similar problem to ours was considered.
	While in that work the focus was on several affine sets, and a threshold value $t_{0}$ was given for each of them by means of their semidefinite programs where observables $\{\A^{t\vx}, \A^{t\vy}, \A^{t\vz}\}$ become compatible, we are considereding \emph{all} affine sets with dimension 2.
\end{remark}

\subsection*{Proof of Proposition \ref{prop:qubit-threshold}: Part 1}

In order to prove that
$M$ is nonempty, let us introduce some relevant notions: 
\begin{align*}
	&D:= \{\mathbf{v}\mid |\mathbf{v}| \leq 1,\ v_z=0\}\subset B:= \{\mathbf{v}\mid|\mathbf{v}| \leq 1\},\\
	&\state_D:= \{\varrho^{\mathbf{v}}\mid
	\mathbf{v} \in D\}\subset\state=\{\varrho^{\mathbf{v}}\mid
	\mathbf{v} \in B\},
\end{align*}
where $\mathbf{v}=v_{\x}\vx+v_{\y}\vy+v_{\z}\vz\in\real^3$, and $\varrho^{\mathbf{v}}:=\frac{1}{2}(\id+\mathbf{v}\cdot\sigma)$.
Since $\state_D$ is a convex set, we can treat $\state_D$ almost like a quantum system.
In the following, we will do it without giving precise definitions because they are obvious. 
For an observable $\E$ on $\mathcal{S}$ with effects $\{\E(x)\}_{x}$, 
we write its restriction 
to $\state_D$ as $\E|_D$ with effects $\{\E(x)|_D\}_{x}$, which is an observable on $\state_D$.
It is easy to obtain the following Lemma. 
\begin{lemma}
	The followings are equivalent:
	\begin{itemize}
		\item[(i)]$\A^{t\vx}$ and $\A^{t\vy}$ are 
		incompatible (thus $\frac{1}{\sqrt{2}}<t\leq 1$).
		\item[(ii)] 
		$\A^{t\vx}$ and $\A^{t\vy}$
		are $\state_D$-incompatible.
		\item[(iii)]
		$\A^{t\vx}|_D$ and $\A^{t\vy}|_D$ 
		are incompatible as observables on $\state_D$. 
	\end{itemize}
\end{lemma}
\begin{proof}
	(i) $\Rightarrow$ (iii). 
	Suppose that $\A^{t\vx}|_D$ and $\A^{t\vy}|_D$ are compatible in $\state_D$. 
	There exists an observable $\M$ on $\state_D$ 
	whose marginals coincide with 
	$\A^{t\vx}|_D$ and $\A^{t\vy}|_D$. 
	One can extend this $\M$ to 
	the whole $\state$ so that it does not 
	depend on $\z$ (for example, one can simply regard its effect $c_{0}\id+c_{1}\sigma_{1}+c_{2}\sigma_{2}$ as an effect on $\state$). Since both $\A^{t\vx}|_D$ and $\A^{t\vy}|_D$ also do not depend on $\z$, the extension of $\M$ gives a joint observable of $\A^{t\vx}$ and $\A^{t\vx}$. 
	\\
	(iii) $\Rightarrow $ (ii). 
	Suppose that $\A^{t\vx}$ and 
	$\A^{t\vy}$ are $\state_D$-compatible. 
	There exists an observable $\M$ on $\state$ 
	whose marginals coincide with $\A^{t\vx}$ 
	and $\A^{t\vy}$ in $\state_D$. 
	The restriction of $\M$ on $\state_D$ 
	proves that (iii) is false. 
	\\
	(ii) $\Rightarrow $ (i). 
	Suppose that $\A^{t\vx}$ and 
	$\A^{t\vy}$ are compatible,
	then they are $\state_D$-compatible. 
\end{proof}
This lemma demonstrates that the incompatibility of $\A^{t\vx}$ and $\A^{t\vy}$ means the incompatibility of $\A^{t\vx}|_{D}$ and $\A^{t\vy}|_{D}$.
We can present further observations.
\begin{lemma}
	Let us consider two pure states 
	$\varrho^{\mathbf{r}_1}$ and $\varrho^{\mathbf{r}_2}$ 
	($\mathbf{r}_1, \mathbf{r}_2 \in \partial B$, $\mathbf{r}_1\neq\mathbf{r}_2$), 
	and a convex subset $\mathcal{S}_0$ of $\state$ generated by them:
	$\mathcal{S}_0:= \{p\varrho^{\mathbf{r}_1}
	+ (1-p)\varrho^{\mathbf{r}_2}\mid0 \leq p \leq 1\}$.
	We also introduce an affine projection $P$ by $P\varrho^{\mathbf{v}}=\varrho_{\P \mathbf{v}}$, where $\varrho^{\mathbf{v}} \in \state$ with $\mathbf{v}=v_x \vx
	+ v_y \vy + v_z \vz$ and $\P\mathbf{v}=v_x \vx
	+ v_y \vy$, and extend it affinely. 
	The affine hull of $\mathcal{S}_0$ is projected to $\state_{D}$ as
	\begin{eqnarray}
		P\state_0:= 
		\{\lambda P \varrho^{\mathbf{r}_{1}} + 
		(1-\lambda) P \varrho^{\mathbf{r}_{2}}\mid\lambda \in \mathbf{R}\}
		\cap \state_D. 
	\end{eqnarray}
	If $\A^{t\vx}$ and $\A^{t\vy}$ are $\mathcal{S}_0$-incompatible, 
	then their restrictions
	$\A^{t\vx}|_D$ and $\A^{t\vy}|_D$ are 
	$P\mathcal{S}_0$-incompatible.  
\end{lemma}
\begin{proof}
	Suppose that $\A^{t\vx}$ and $\A^{t\vy}$ are $\mathcal{S}_0$-incompatible. 
	It implies $\P \mathbf{r}_{1}\neq\P \mathbf{r}_{2}$ i.e. $P \varrho^{\mathbf{r}_{1}}\neq P \varrho^{\mathbf{r}_{2}}$ (see Example \ref{ex:fix}), and thus $P\state_0$ is a segment in $\state_{D}$.

	If $\A^{t\vx}|_D$ and $\A^{t\vy}|_D$ 
	are $P\state_0$-compatible, then there exists a joint observable $\M$ 
	on $\state_D$ such that 
	its marginals coincide with $\A^{t\vx}|_D$ 
	and $\A^{t\vy}|_D$
	on $P\state_0 \subset \state_D$. 
	This $\M$ can be extended to an observable on
	$\state$ so that the extension does not depend on 
	$\z$.
	Because 
	\begin{align*}
		&\tr{\A^{t\vx}(\pm)P\varrho^{\mathbf{r}_{1}}}=\tr{\A^{t\vx}(\pm)\varrho^{\mathbf{r}_{1}}},\\
		&\tr{\A^{t\vx}(\pm)P\varrho^{\mathbf{r}_{2}}}=\tr{\A^{t\vx}(\pm)\varrho^{\mathbf{r}_{2}}}
	\end{align*}
	(and their $\vy$-counterparts) hold due to the independence of $\A^{t\vx}(\pm)$ from $\sigma_3$, the marginals of $\M$ coincide with $\A^{t\vx}$ and $\A^{t\vy}$ on $\state_0$.
	It results in the $\mathcal{S}_0$-compatibility of $\A^{t\vx}$ 
	and $\A^{t\vy}$,
	which is a contradiction.
\end{proof}
It follows from this lemma that $\chi_{incomp}(\A^{t\vx}|_{D},\A^{t\vy}|_{D})$ is two when $\chi_{incomp}(\A^{t\vx},\A^{t\vy})$ is two, equivalently $\chi_{incomp}(\A^{t\vx},\A^{t\vy})$ is three when $\chi_{incomp}(\A^{t\vx}|_{D},\A^{t\vy}|_{D})$ is three (remember that $\chi_{incomp}(\A^{t\vx},\A^{t\vy})\le3$).
In fact, the converse also holds.
\begin{lemma}
	$\chi_{incomp}(\A^{t\vx}|_{D},\A^{t\vy}|_{D})$ is three
	when $\chi_{incomp}(\A^{t\vx}, \A^{t\vy})$ is three.
\end{lemma}
\begin{proof}
	Let $\chi_{incomp}(\A^{t\vx},\A^{t\vy})=3$.
	It follows that for any line $S\subset\state$, $\A^{t\vx}$ and $\A^{t\vy}$ are $S$-compatible.
	In particular, $\A^{t\vx}$ and $\A^{t\vy}$ are $S'$-compatible for any line $S'$ in $\state_{D}$, and thus there is an observable $\M$ such that its marginals coincide with $\A^{t\vx}$ and $\A^{t\vy}$ on $S'$.
	It is easy to see that the marginals of $\M|_{D}$ coincide with $\A^{t\vx}|_{D}$ and $\A^{t\vy}|_{D}$ on $S'$, which results in the $S'$-compatibility of $\A^{t\vx}|_{D}$ and $\A^{t\vy}|_{D}$.
	Because $S'$ is arbitrary, we can conclude $\chi_{incomp}(\A^{t\vx}|_{D},\A^{t\vy}|_{D})=3$.
\end{proof}

The lemmas above manifest that if $\A^{t\vx}$ and $\A^{t\vy}$ are incompatible, then $\A^{t\vx}|_{D}$ and $\A^{t\vy}|_{D}$ are also incompatible and 
\[
\chi_{incomp}(\A^{t\vx},\A^{t\vy})=\chi_{incomp}(\A^{t\vx}|_{D},\A^{t\vy}|_{D}).
\]
Therefore, in the following, we denote $\A^{t\vx}|_{D}$ and $\A^{t\vy}|_{D}$ simply by $\A^{t\vx}_{D}$ and $\A^{t\vy}_{D}$ respectively, and focus on the quantity $\chi_{incomp}(\A^{t\vx}_{D},\A^{t\vy}_{D})$ instead of $\chi_{incomp}(\A^{t\vx}, \A^{t\vy})$.
Before proceeding to the next step, let us confirm our strategy of this part.
It is composed by further two parts: \textbf{(a)} and \textbf{(b)}. 
In \textbf{(a)}, we will consider a line (segment) $\state_{1}$ in $\state_{D}$, and consider for $0<t<1$ all pairs of observables $(\widetilde{\A}_{1}^{t}, \widetilde{\A}_{2}^{t})$ on $\state_{D}$ which coincide with $(\A^{t\vx}_{D}, \A^{t\vy}_{D})$ on $\state_{1}$.
Then, in \textbf{(b)}, we will investigate the (in)compatibility of those $\widetilde{\A}_{1}^{t}$ and $\widetilde{\A}_{2}^{t}$ in order to obtain $\chi_{incomp}(\A^{t\vx}_{D},\A^{t\vy}_{D})$. 
It will be shown that when $t$ is sufficiently small, there exists a compatible pair $(\widetilde{\A}_{1}^{t}, \widetilde{\A}_{2}^{t})$ for any $\state_{1}$, that is, $\A^{t\vx}_{D}$ and $\A^{t\vy}_{D}$ are $\state_{1}$-compatible for any line $\state_{1}$.
It results in $\chi_{incomp}(\A^{t\vx}_{D},\A^{t\vy}_{D})=3$, and thus $M\neq \emptyset$.

\textbf{(a)}\quad
Let us consider two pure states 
$\varrho^{\mathbf{r}_1}$ and $\varrho^{\mathbf{r}_{2}}$ 
with $\mathbf{r}_1, \mathbf{r}_2 \in \partial D$ ($\mathbf{r}_{1}\neq\mathbf{r}_{2}$), 
and a convex set 
$\state_1:= \{p \varrho^{\mathbf{r}_1}
+ (1-p) \varrho^{\mathbf{r}_2}\mid 0 \leq p \leq 1\}$. 
We set parameters $\varphi_1$ and $\varphi_2$ 
as 
\begin{align}
	&\mathbf{r}_1= \cos \varphi_1 \vx
	+ \sin \varphi_1 \vy, \\
	&\mathbf{r}_2 = \cos \varphi_2 \vx 
	+ \sin \varphi_2 \vy, 
\end{align}
where $-\pi\le\varphi_1<\varphi_2<\pi$.
By exchanging $\pm$ properly, without loss of 
generality
we can assume 
the line connecting $\mathbf{r}_1$ and 
$\mathbf{r}_2$ passes through above the origin
(instead of below).
In this case, from geometric consideration, we have 
\begin{equation}
	\label{eq:bound for varphi 0}
	\begin{aligned}
		&0<\varphi_2 - \varphi_1 \leq \pi, \\
		&0 \le \frac{\varphi_1 + \varphi_2}{2} \le \frac{\pi}{2}. 
	\end{aligned}
\end{equation}
Note that when $\varphi_2 - \varphi_1=\pi$, $\varrho^{\mathbf{r}_{1}}$ and $\varrho^{\mathbf{r}_{2}}$ are perfectly distinguishable, which results in the $\state_{1}$-compatibility of $\A^{t\vx}_{D}$ and $\A^{t\vy}_{D}$ (see Example \ref{ex:distinguishable}).
On the other hand, when $\frac{\varphi_1 + \varphi_2}{2}=0$ or $\frac{\pi}{2}$,  $\tr{\varrho\A^{t\vx}_{D}
	(+)}$ or $\tr{\varrho\A^{t\vy}_{D}
	(+)}$ is constant for $\varrho\in\state_1$ respectively, so $\A^{t\vx}_D$ and $\A^{t\vy}_D$ are $\state_1$-compatible (see Example \ref{ex:fix}). 
Thus, instead of \eqref{eq:bound for varphi 0}, we hereafter assume 
\begin{equation}
	\label{eq:bound for varphi }
	\begin{aligned}
		&0<\frac{\varphi_2 - \varphi_1}{2} < \frac{\pi}{2}, \\
		&0 < \frac{\varphi_1 + \varphi_2}{2} < \frac{\pi}{2}.
	\end{aligned}
\end{equation}
Next, we consider a binary observable $\widetilde{\A}_{1}^{t}$
on $\state_D$ which coincides with $\A^{t\vx}_{D}$ on $\state_1\subset\state_{D}$.
There are many possible $\widetilde{\A}_{1}^{t}$, and each $\widetilde{\A}_{1}^{t}$ is determined completely by its effect $\widetilde{\A}_{1}^{t}(+)$ corresponding to the outcome `+' because it is binary.
The effect $\widetilde{\A}_{1}^{t}(+)$ is associated with a vector $\mathbf{v}_1\in D$ defined as
\begin{eqnarray}
	\mathbf{v}_1:=
	argmax_{\mathbf{v}\in D} \mbox{tr}[\varrho_{\mathbf{v}} 
	\widetilde{\A}_{1}^{t}(+)].
\end{eqnarray} 
Let us introduce a parameter $\xi_1\in[-\pi, \pi)$ by
\begin{eqnarray}
	\label{eq:E1_v1}
	\mathbf{v}_1 = \cos \xi_{1}
	\vx + \sin \xi_{1} \vy,
\end{eqnarray}
and express $\widetilde{\A}_{1}^{t}(+)$ as
\begin{eqnarray}
	\label{eq:E1_Bloch}
	\widetilde{\A}_{1}^{t}(+) 
	= \frac{1}{2}\left(
	(1+w(\xi_{1})) \id + \mathbf{m}_1(\xi_1)
	\cdot \mathbf{\sigma}\right),
\end{eqnarray}
where we set
\begin{equation}
	\label{eq:E1_A1}
	\mathbf{m}_1(\xi_1) = C_1(\xi_1) \mathbf{v}_1\quad\mbox{with}\quad0\le C_1(\xi_1)\le1.
\end{equation}
Because 
\begin{align*}
	&\tr{\varrho^{\mathbf{r}_{1}}\A^{t\vx}_{D} (+)}=\tr{\varrho^{\mathbf{r}_{1}}\widetilde{\A}_{1}^{t}(+)},\\
	&\tr{\varrho^{\mathbf{r}_{2}}\A^{t\vy}_{D} (+)}=\tr{\varrho^{\mathbf{r}_{2}}\widetilde{\A}_{1}^{t}(+)},
\end{align*}
namely
\begin{equation}
	\label{eq:E1}
	\begin{aligned}
		&\frac{1}{2}+\frac{t}{2}\cos\varphi_{1}=\frac{1+w_1(\xi_1)}{2}+\frac{C_{1}(\xi_{1})}{2}\cos(\varphi_{1}-\xi_{1}),\\
		&\frac{1}{2}+\frac{t}{2}\cos\varphi_{2}=\frac{1+w_1(\xi_1)}{2}+\frac{C_{1}(\xi_{1})}{2}\cos(\varphi_{2}-\xi_{1}),
	\end{aligned}
\end{equation}
hold, we can obtain
\begin{align}
	\label{eq:A_1explicit}
	C_1(\xi_1)
	&= \frac{t(\cos \varphi_1 
		- \cos \varphi_2)}{\cos (\varphi_1 - \xi_1)
		-\cos (\varphi_2 - \xi_1)}\notag\\
	&= \frac{t\sin\varphi_{0}}{
		\sin (\varphi_{0}- \xi_1)},\\
	\label{eq:x_1explicit}
	w_1(\xi_1) 
	&= - t \left(
	\frac{\sin (\varphi_1 - \varphi_2)}{
		2 \sin(\frac{\varphi_1 - \varphi_2}{2})}\right)
	\times \left(\frac{\sin \xi_1}{\sin 
		(\varphi_{0} -\xi_1)} \right)\notag\\
	&= 
	\frac{-t\cos\psi_{0}\sin \xi_1}{\sin 
		(\varphi_{0} -\xi_1)},
\end{align}
where we set $\varphi_{0}:=\frac{\varphi_{1}+\varphi_{2}}{2}$ and $\psi_{0}:=\frac{\varphi_{2}-\varphi_{1}}{2}$ ($0<\varphi_{0}<\frac{\pi}{2}$, $0<\psi_{0}<\frac{\pi}{2}$).
Note that if $\sin (\varphi_{0}- \xi_1)=0$ or $\cos (\varphi_1 - \xi_1)
-\cos (\varphi_2 - \xi_1)=0$ holds, then $\cos \varphi_1 
- \cos \varphi_2=0$ holds (see \eqref{eq:A_1explicit}). It means $\varphi_{0}=0$, which is a contradiction, and thus $\sin (\varphi_{0}- \xi_1)\neq0$ (that is, $C_{1}(\xi_{1})$ and $w_{1}(\xi_{1})$ in \eqref{eq:A_1explicit}, \eqref{eq:x_1explicit} are well-defined).
Moreover, because $C_{1}(\xi_{1})\ge0$, we can see from \eqref{eq:A_1explicit} that $\sin (\varphi_{0}- \xi_1)>0$ holds, which results in 
\begin{align}
	0\le\xi_{1}<\varphi_{0}\label{eq:xi_1max},
\end{align}
or
\begin{align}
	-\pi+\varphi_{0}<\xi_{1}\le0\label{eq:xi_1min}.
\end{align}
In addition, $\xi_1$ is restricted also by the condition that $\widetilde{\A}_{1}^{t}(\pm)$ are positive.
Since the eigenvalues of $\widetilde{\A}_{1}^{t}(\pm)$ are $\frac{1}{2}((1+w_1(\xi_1)) \pm C_1(\xi_1))$, the restriction comes from both 
\begin{equation}
	\begin{aligned}
		&1+w_1(\xi_1) +C_1(\xi_1) \leq 2,
		\\
		&1+w_1(\xi_1) - C_1(\xi_1) \geq 0, 
	\end{aligned}
\end{equation}
equivalently
\begin{align}
	&1-w_1(\xi_1) \geq C_1(\xi_1)\label{condition:min}, 
	\\
	&1+w_1(\xi_1) \geq C_1(\xi_1)\label{condition:max}.
\end{align}
When \eqref{eq:xi_1min} (i.e. $\sin\xi_{1}\le0$) holds, $w_{1}(\xi_{1})\ge0$ holds, and thus \eqref{condition:min} is sufficient. 
It is written explicitly as
\begin{align*}
	\sin\left(\varphi_{0}-\xi_{1}\right)+t\sin\xi_{1}\cos\psi_{0}
	\ge t\sin\varphi_{0},
\end{align*}
or
\begin{align}
	\label{eq:E1_innerproduct1}
	\frac{1}{t}\cos\xi_{1}+	\frac{1}{t\sin\varphi_{0}}\left(t\cos\psi_{0}-\cos\varphi_{0}\right)\sin\xi_{1}
	\ge
	1.
\end{align}
In order to investigate \eqref{eq:E1_innerproduct1}, we adopt a geometric method here while it can be solved in an analytic way.
Let us define 
\begin{equation}
	\label{eq:def of h1}
	h_{1}(t, \varphi_{0}, \psi_{0})=\frac{1}{t\sin\varphi_{0}}\left(t\cos\psi_{0}-\cos\varphi_{0}\right).
\end{equation}
Then, we can rewrite \eqref{eq:E1_innerproduct1} as
\begin{equation}
	\label{eq:E1_innerproduct2}
	(\cos\xi_{1}, \sin\xi_{1})
	\cdot
	\left[\left(\frac{1}{t}, h_{1}\right)-(\cos\xi_{1}, \sin\xi_{1})\right]\ge0.
\end{equation}
In fact, it can be verified easily that $\left(\frac{1}{t}, h_{1}\right)$ is the intersection of the line $l_{1}:=\{\lambda\mathbf{r}_1+(1-\lambda)\mathbf{r}_2\mid \lambda\in\real\}$ and the line $x=\frac{1}{t}$ in $\real^{2}$.
Considering this fact, we can find that $\xi_{1}$ satisfies \eqref{eq:E1_innerproduct2} if and only if
\begin{equation}
	\xi_{1}^{min}(t, \varphi_{0}, \psi_{0})\le\xi_{1}\le0,
\end{equation}
where $\xi_{1}^{min}(t, \varphi_{0}, \psi_{0})$ is determined by the condition 
\begin{align}
	\left[\left(\frac{1}{t}, h_{1}\right)-(\cos\xi_{1}^{min}, \sin\xi_{1}^{min})\right]\perp(\cos\xi_{1}^{min}, \sin\xi_{1}^{min})
\end{align}
(see FIG. \ref{Fig:xi_min}). 
\begin{figure}[h]
		\includegraphics[bb=0.000000 0.000000 786.000000 630.000000, scale=0.3]{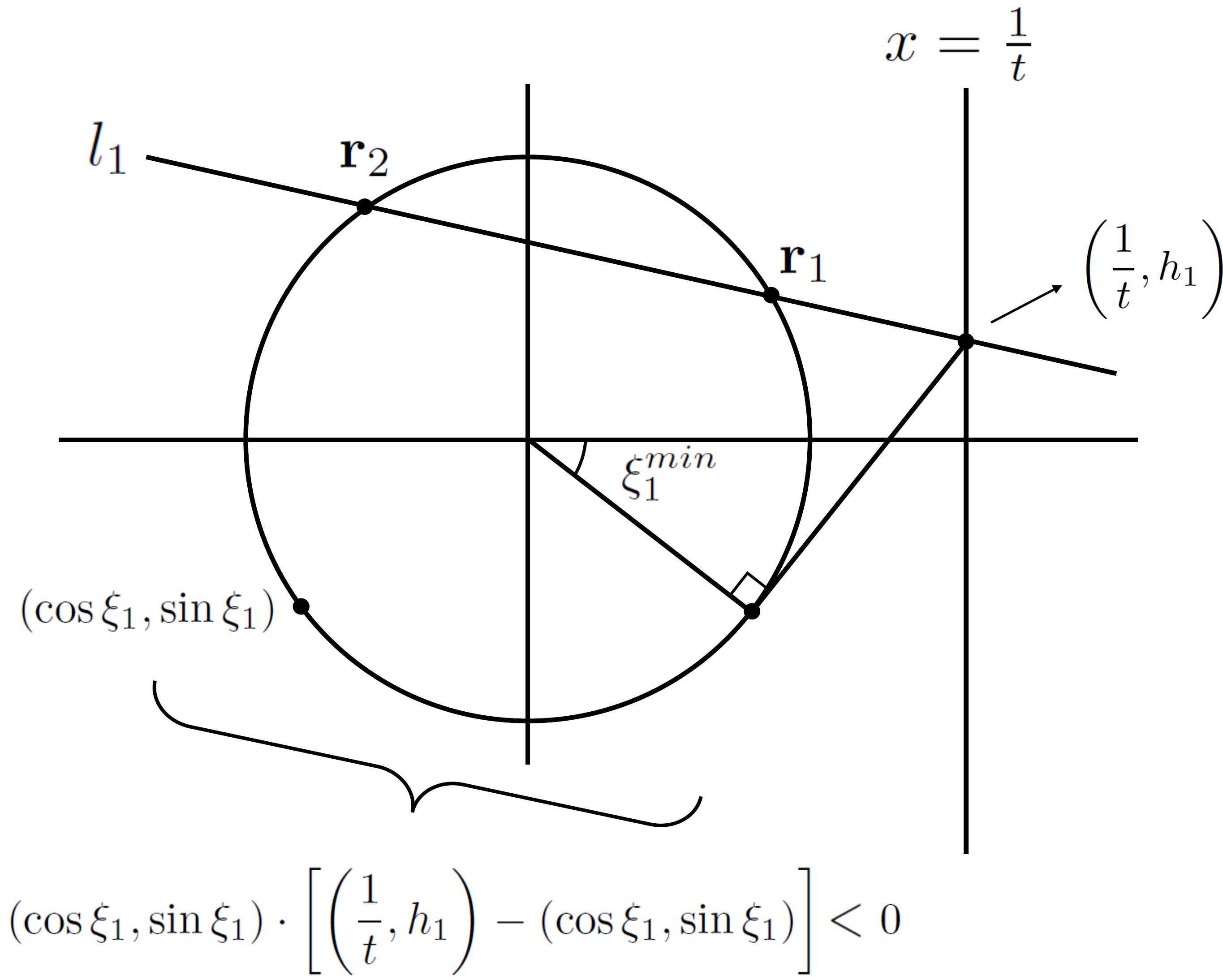}
	\caption{Geometric description of determining $\xi_{1}^{min}$.}
	\label{Fig:xi_min}
\end{figure}
Analytically, it corresponds to the case when the equality of \eqref{eq:E1_innerproduct1} holds:
\begin{align}
	\label{eq:xi_1min_identity}
	\frac{1}{t}\cos\xi_{1}^{min}+	\frac{1}{t\sin\varphi_{0}}\left(t\cos\psi_{0}-\cos\varphi_{0}\right)\sin\xi_{1}^{min}
	=
	1,
\end{align}
or
\[
1-w_{1}(\xi_{1}^{min})=C_{1}(\xi_{1}^{min}).
\]
It can be represented explicitly as
\begin{equation}
	\label{eq:xi1_equation}
	\begin{aligned}
		\left( 
		t^2 \cos^2\psi_{0}
		-2t \cos \varphi_{0}\cos \psi_{0}
		+1\right)
		\sin^2\xi_1^{min}\\
		- 2t \sin\varphi_{0}
		\left( t \cos\psi_{0}
		-\cos \varphi_{0}\right) \sin \xi_1^{min}\\ 
		+(t^2 -1) \sin^2 \varphi_{0}=0,
	\end{aligned}
\end{equation}
and $\sin\xi_{1}^{min}$ is obtained as its negative solution.
Note that since the coefficient $(t^2 \cos^2\psi_{0}-
2t \cos \varphi_{0}\cos \psi_{0}
+1)$ is strictly positive, the solutions do not show any singular behavior.
In summary, we have obtained
\begin{equation}
	\xi_{1}^{min}(t, \varphi_{0}, \psi_{0})\le\xi_{1}\le0
\end{equation}
with $\xi_{1}^{min}(t, \varphi_{0}, \psi_{0})$ uniquely determined for $t$, $\varphi_{0}$, and $\psi_{0}$ by 
\begin{align}
	\left\{
	\begin{aligned}
		&-\pi+\varphi_{0}<\xi_{1}^{min}(t, \varphi_{0}, \psi_{0})\le0,\\
		&1-w_{1}(\xi_{1}^{min}(t, \varphi_{0}, \psi_{0}))=C_{1}(\xi_{1}^{min}(t, \varphi_{0}, \psi_{0})).
	\end{aligned}
	\right.
	\label{eq:xi_1min_detailed}
\end{align}
On the other hand, when \eqref{eq:xi_1max} (i.e. $\sin\xi_{1}\ge0$) holds, \eqref{condition:max} is sufficient. 
It results in a tight condition for $\xi_{1}$:
\begin{align}
	0\le\xi_{1}\le\xi_{1}^{max}(t, \varphi_{0}, \psi_{0}),
\end{align}
where $\xi_{1}^{max}(t, \varphi_{0}, \psi_{0})$ is a constant uniquely determined for $\varphi_{0}$ and $\psi_{0}$ by 
\begin{align}
	\left\{
	\begin{aligned}
		&0\le\xi_{1}^{max}(t, \varphi_{0}, \psi_{0})<\varphi_{0}\\
		&1+w_{1}(\xi_{1}^{max}(t, \varphi_{0}, \psi_{0}))=C_{1}(\xi_{1}^{max}(t, \varphi_{0}, \psi_{0})).
	\end{aligned}
	\right.
	\label{eq:xi_1max_detailed}
\end{align}
We remark that this can be obtained by a similar geometric method to the previous case: consider the intersection of the line $l_{1}$ and the line $x=-\frac{1}{t}$ in turn (see FIG. \ref{Fig:xi_max}).
\begin{figure}[h]
	\includegraphics[bb=0.000000 0.000000 660.000000 507.000000, scale=0.31]{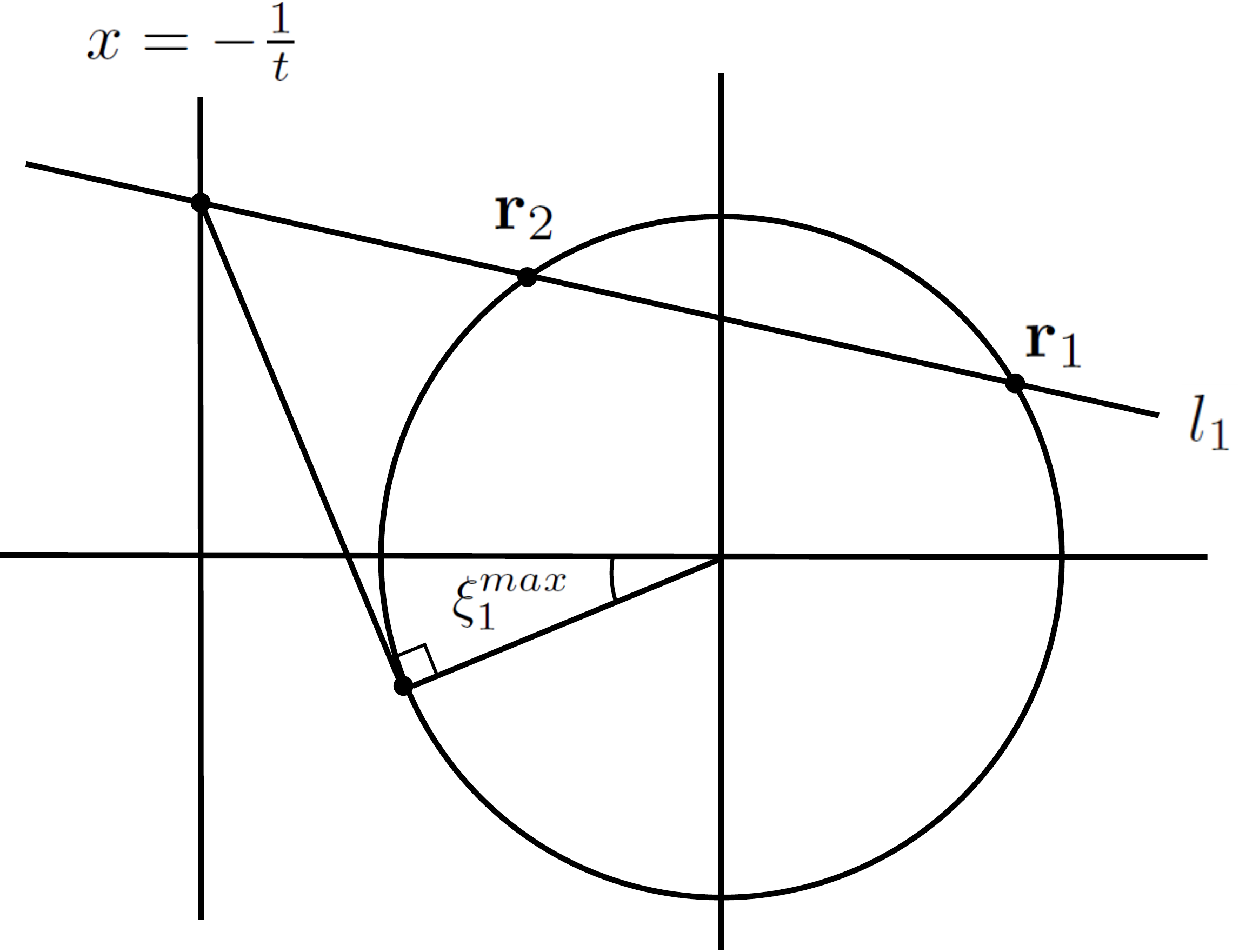}
	\caption{Geometric description of determining $\xi_{1}^{max}$.}
	\label{Fig:xi_max}
\end{figure}
Overall, we have demonstrated that $\xi_{1}$ for $\widetilde{\A}_{1}^{t}$ satisfies
\begin{align}
	\xi_{1}^{min}(t, \varphi_{0}, \psi_{0})\le\xi_{1}\le\xi_{1}^{max}(t, \varphi_{0}, \psi_{0}),
\end{align}
where $\xi_{1}^{min}(t, \varphi_{0}, \psi_{0})$ and $\xi_{1}^{max}(t, \varphi_{0}, \psi_{0})$ are obtained thorough \eqref{eq:xi_1min_detailed} and \eqref{eq:xi_1max_detailed} respectively.
Note that $\xi_{1}^{min}(t, \varphi_{0}, \psi_{0})$ and $\xi_{1}^{max}(t, \varphi_{0}, \psi_{0})$ depend
continuously on $t$ (and $\varphi_1, \varphi_2$ through 
$\varphi_0$ and $\psi_{0}$).

Similarly, we consider a binary observable $\widetilde{\A}_{2}^{t}$ on $\state_D$ which coincides with $\A_{D}^{t\vy}$ in $\state_1$, and focus on its effect $\widetilde{\A}_{2}^{t}(+)$.
We define parameters $\mathbf{v}_2\in D$ and $\xi_2\in[-\pi, \pi)$ as
\begin{equation}
	\begin{aligned}
		\mathbf{v}_2 
		&= \sin \xi_2 \vx +\cos \xi_2 \vy \\
		&= argmax_{\mathbf{v}\in D} \mbox{tr}
		[\widetilde{\A}_{2}^{t} (+) \varrho_{\mathbf{v}}]. 
	\end{aligned}
\end{equation}
$\widetilde{\A}_{2}^{t}(+)$ is represented as
\begin{eqnarray}
	\label{eq:E2_Bloch}
	\widetilde{\A}_{2}^{t}(+) = \frac{1}{2}
	\left( (1+w_2(\xi_2))\id + \mathbf{m}_2(\xi_2)\right)
\end{eqnarray}
with
\[
\mathbf{m}_2(\xi_2) = C_2(\xi_2) \mathbf{v}_2\quad(0\le C_2(\xi_2)\le1).
\]
\eqref{eq:E1} becomes
\begin{equation}
	\begin{aligned}
		&\frac{1}{2}+\frac{t}{2}\cos\left(\frac{\pi}{2}-\varphi_{1}\right)\\
		&\qquad=\frac{1+x_2(\xi_2)}{2}+\frac{C_{2}(\xi_{2})}{2}\cos\left(\frac{\pi}{2}-\varphi_{1}-\xi_{1}\right),\\
		&\frac{1}{2}+\frac{t}{2}\cos\left(\frac{\pi}{2}-\varphi_{2}\right)\\
		&\qquad=\frac{1+x_2(\xi_2)}{2}+\frac{C_{2}(\xi_{2})}{2}\cos\left(\frac{\pi}{2}-\varphi_{2}-\xi_{1}\right),
	\end{aligned}
\end{equation}
so defining $\overline{\varphi_{1}}:=\frac{\pi}{2}-\varphi_{1}$ and $\overline{\varphi_{2}}:=\frac{\pi}{2}-\varphi_{2}$, we can obtain similarly to \eqref{eq:A_1explicit} and \eqref{eq:x_1explicit}
\begin{align}
	\label{eq:A_2explicit}
	&C_2(\xi_2)
	= \frac{t\sin\overline{\varphi_{0}}}{
		\sin (\overline{\varphi_{0}}- \xi_1)}
	, \\
	\label{eq:x_2explicit}
	&w_2(\xi_2) = \frac{-t\cos\psi_{0}\sin \xi_2}{\sin 
		(\overline{\varphi_{0}}-\xi_2)},
\end{align}
where $\overline{\varphi_{0}}:=\frac{\overline{\varphi_{1}}+\overline{\varphi_{2}}}{2}=\frac{\pi}{2}-\varphi_{0}$.
It follows that properties of $\widetilde{\A}_{2}^{t}$ can be obtained just by replacing $\xi_{1}$ and $\varphi_{0}$ exhibited in the argument for $\widetilde{\A}_{1}^{t}$ by $\xi_{2}$ and $\overline{\varphi_{0}}$ respectively.  
Remark that $0<\overline{\varphi_{0}}<\frac{\pi}{2}$ holds similarly to $\varphi_{0}$, and that the change $\psi_{0}\rightarrow\overline{\psi_{0}}:=\frac{\overline{\varphi_{2}}-\overline{\varphi_{1}}}{2}=-\psi_{0}$ does not affect the equations above, so we dismiss it.
From \eqref{eq:A_2explicit} and \eqref{eq:x_2explicit}, we have 
\begin{align}
	\xi_{2}^{min}(t, \varphi_{0}, \psi_{0})
	\le
	\xi_{2}
	\le
	\xi_{2}^{max}(t, \varphi_{0}, \psi_{0}),
\end{align}
where 
\begin{equation}
	\begin{aligned}
		\xi_{2}^{min}(t, \varphi_{0}, \psi_{0})
		&=\xi_{1}^{min}(t, \overline{\varphi_{0}}, \psi_{0})\\
		&=\xi_{1}^{min}\left(t, \frac{\pi}{2}-\varphi_{0}, \psi_{0}\right),
	\end{aligned}
\end{equation}
and
\begin{equation}
	\begin{aligned}
		\xi_{2}^{max}(t, \varphi_{0}, \psi_{0})
		&=\xi_{1}^{max}(t, \overline{\varphi_{0}}, \psi_{0})\\
		&=\xi_{1}^{max}\left(t, \frac{\pi}{2}-\varphi_{0}, \psi_{0}\right),
	\end{aligned}
\end{equation}
which satisfy
\begin{align}
	\left\{
	\begin{aligned}
		&-\frac{\pi}{2}+\varphi_{0}<\xi_{2}^{min}(t, \varphi_{0}, \psi_{0})
		\le0\\
		&1-w_{2}(\xi_{2}^{min}(t, \varphi_{0}, \psi_{0}))=C_{2}(\xi_{2}^{min}(t, \varphi_{0}, \psi_{0}))
	\end{aligned}
	\right.
	\label{eq:xi_2min_detailed}
\end{align}
and
\begin{align}
	\left\{
	\begin{aligned}
		&0\le
		\xi_{2}^{max}(t, \varphi_{0}, \psi_{0})<
		\frac{\pi}{2}-\varphi_{0}\\
		&1+w_{2}(\xi_{2}^{max}(t, \varphi_{0}, \psi_{0}))=C_{2}(\xi_{2}^{max}(t, \varphi_{0}, \psi_{0}))
	\end{aligned}
	\right.
	\label{eq:xi_2max_detailed}
\end{align}
respectively.\\

\textbf{(b)}\quad
In this part, we shall consider the (in)compatibility of the observables $\widetilde{\A}_{1}^{t}$ and $\widetilde{\A}_{2}^{t}$ defined in $\textbf{(a)}$ for $t$ close to $\frac{1}{\sqrt{2}}$ ($t\sim\frac{1}{\sqrt{2}}$).
It is related directly with the $\state_{1}$-(in)compatibility of $\A^{t\vx}_{D}$ and $\A^{t\vy}_{D}$ as we have shown in the beginning of this section.

Let us examine the behavior of $\xi_{1}^{min}(t, \varphi_{0}, \psi_{0})$ for $t\sim\frac{1}{\sqrt{2}}$.
We denote $\xi_{1}^{min}(t=\frac{1}{\sqrt{2}}, \varphi_{0}, \psi_{0})$ and $h_{1}(t=\frac{1}{\sqrt{2}}, \varphi_{0}, \psi_{0})$ simply by $\widehat{\xi}_{1}^{min}(\varphi_{0}, \psi_{0})$ and $\widehat{h}_{1}(\varphi_{0}, \psi_{0})$ respectively.
The following lemma is useful.
\begin{lemma}
	\label{lem:decrease}
	With $\varphi_{0}$ fixed, $\widehat{\xi}_{1}^{min}$ is a strictly decreasing function of $\psi_{0}$.
\end{lemma}
\begin{proof}
	The claim can be observed to hold by a geometric consideration in terms of FIG. \ref{Fig:xi_min}. 
	In fact, increasing $\psi_0$ with $\varphi_0$ fixed corresponds to moving the line $l_1$ down with its inclination fixed. 
	The movement makes $h_1$ (or $\widehat{h}_{1}$) and hence $\xi_1^{min}$ (or $\widehat{\xi}_{1}^{min}$) smaller, which proves the claim.
	Here we show an analytic proof of this fact.
	We can see from \eqref{eq:def of h1} and \eqref{eq:xi_1min_identity} that 
	\begin{align}
		\label{eq:xi1_hat}
		\sqrt{2}\cos\widehat{\xi}_{1}^{min}+	\widehat{h}_{1}\sin\widehat{\xi}_{1}^{min}
		=
		1,
	\end{align}
	i.e.
	\[
	\widehat{h}_{1}=\frac{1}{\sin\widehat{\xi}_{1}^{min}}\left(1-\sqrt{2}\cos\widehat{\xi}_{1}^{min}\right)
	\]
	holds  
	(note that $\sin\widehat{\xi}_{1}^{min}\neq0$ because $\sin\widehat{\xi}_{1}^{min}=0$ contradicts \eqref{eq:xi1_hat}). 	Because
	\[
	\frac{d\widehat{h}_{1}}{d\widehat{\xi}_{1}^{min}}=\frac{1}{(\sin\widehat{\xi}_{1}^{min})^{2}}\left(\sqrt{2}-\cos\widehat{\xi}_{1}^{min}\right)>0,
	\]
	and $\widehat{h}_{1}=\frac{1}{\sin\varphi_{0}}\left(\cos\psi_{0}-\sqrt{2}\cos\varphi_{0}\right)$ is a decreasing function of $\psi_{0}$, the claim follows.
\end{proof}
From this lemma, it follows that 
\begin{align}
	\label{eq:Xi1_def}
	\widehat{\xi}_{1}^{min}(\varphi_{0}, \psi_{0})<
	\lim_{\psi_{0}\rightarrow +0}\widehat{\xi}_{1}^{min}(\varphi_{0}, \psi_{0})=:{\Xi}_{1}^{min}(\varphi_{0}),
\end{align} 
and
\begin{align}
	\label{eq:Xi2_def}
	\widehat{\xi}_{2}^{min}(\varphi_{0}, \psi_{0})<{\Xi}_{2}^{min}(\varphi_{0})
\end{align}
hold for all $\varphi_{0}\in (0, \frac{\pi}{2})$ and $\psi_{0}\in (0, \frac{\pi}{2})$, where
\begin{equation}
	\label{eq:Xi2_def0}
	\begin{aligned}
		&\widehat{\xi}_{2}^{min}(\varphi_{0}, \psi_{0}):=\xi_{2}^{min}\left(t=\frac{1}{\sqrt{2}},\varphi_{0}, \psi_{0}\right)\\
		&\qquad\qquad\quad\ \left(=\widehat{\xi}_{1}^{min}\left(\frac{\pi}{2}-\varphi_{0}, \psi_{0}\right)\right),\\
		&{\Xi}_{2}^{min}(\varphi_{0}):=\lim_{\psi_{0}\rightarrow +0}\widehat{\xi}_{2}^{min}(\varphi_{0}, \psi_{0})\\
		&\qquad\qquad\left(=\Xi_{1}^{min}\left(\frac{\pi}{2}-\varphi_{0}\right)\right).
	\end{aligned}
\end{equation}
We can prove the following lemma.
\begin{lemma}
	\label{lem:Xi_bounds}
	\[
	\Xi_{1}^{min}(\varphi_{0})+\Xi_{2}^{min}(\varphi_{0})\le-\frac{\pi}{2}
	\]
	holds for all $0<\varphi_{0}<\frac{\pi}{2}$.
\end{lemma}
\begin{proof}
	Let us define 
	\begin{align*}
		H_{1}(\varphi_{0}):&=\lim_{\psi_{0}\rightarrow +0}\widehat{h}_{1}(\varphi_{0}, \psi_{0})\\
		&=\lim_{\psi_{0}\rightarrow +0}h_{1}\left(t=\frac{1}{\sqrt{2}}, \varphi_{0}, \psi_{0}\right)\\
		&=\frac{1}{\sin\varphi_{0}}\left(1-\sqrt{2}\cos\varphi_{0}\right).
	\end{align*}
	It holds similarly to \eqref{eq:xi1_hat} that
	\begin{align}
		\label{eq:Xi1_eq}
		\sqrt{2}\cos\Xi_{1}^{min}+	H_{1}\sin\Xi_{1}^{min}
		=
		1.
	\end{align}
	Hence, together with $\sin^{2}\Xi_{1}^{min}+\cos^{2}\Xi_{1}^{min}=1$, we can obtain 
	\begin{align}
		\label{eq:Xi1=H1}
		\cos\Xi_{1}^{min}=\frac{1}{\sqrt{2}}\cdot\frac{2+H_{1}\sqrt{2H_{1}^{2}+2}}{H_{1}^{2}+2},
	\end{align}
	or its more explicit form 
	\begin{align}
		\cos\Xi_{1}^{min}=\frac{1}{\sqrt{2}}\cdot\frac{4-3\sqrt{2}\cos\varphi_{0}}{3-2\sqrt{2}\cos\varphi_{0}}.
	\end{align}
	It results in
	\begin{align}
		\label{eq:Xi_arccos}
		\Xi_{1}^{min}(\varphi_{0})=-\arccos\left(\frac{1}{\sqrt{2}}\cdot\frac{4-3\sqrt{2}\cos\varphi_{0}}{3-2\sqrt{2}\cos\varphi_{0}}\right),
	\end{align}
	where we follow the convention that $\arccos\colon[-1, 1]\to[0, \pi]$, and thus ${\Xi}_{1}^{min}\in(-\pi+\varphi_{0}, 0]$ is obtained through $-\arccos\colon[-1, 1]\to[-\pi, 0]$.
	Because
	\[
	\frac{d}{d\varphi_{0}}\left(\frac{1}{\sqrt{2}}\cdot\frac{4-3\sqrt{2}\cos\varphi_{0}}{3-2\sqrt{2}\cos\varphi_{0}}\right)
	=\frac{\sin\varphi_{0}}{(3-2\sqrt{2}\cos\varphi_{0})^{2}},
	\]
	and
	\begin{align*}
		\sqrt{1-\left(\frac{1}{\sqrt{2}}\cdot\frac{4-3\sqrt{2}\cos\varphi_{0}}{3-2\sqrt{2}\cos\varphi_{0}}\right)^{2}}
		&=\sqrt{\left(\frac{\sin\varphi_{0}}{3-2\sqrt{2}\cos\varphi_{0}}\right)^{2}}\\
		&=\frac{\sin\varphi_{0}}{3-2\sqrt{2}\cos\varphi_{0}},
	\end{align*}
	we can observe that
	\begin{align*}
		\frac{d\Xi_{1}^{min}}{d\varphi_{0}}
		&=\left(\frac{\sin\varphi_{0}}{3-2\sqrt{2}\cos\varphi_{0}}\right)^{-1}\cdot\frac{\sin\varphi_{0}}{(3-2\sqrt{2}\cos\varphi_{0})^{2}}\\
		&=\frac{1}{3-2\sqrt{2}\cos\varphi_{0}},
	\end{align*}
	and
	\begin{align}
		\frac{d^{2}\Xi_{1}^{min}}{d\varphi_{0}^{2}}=\frac{-2\sqrt{2}\sin\varphi_{0}}{(3-2\sqrt{2}\cos\varphi_{0})^{2}}<0,
	\end{align}
	which means $\Xi_{1}^{min}$ is concave.
	Therefore, for any $\varphi_{0}\in(0, \frac{\pi}{2})$, the concavity results in 
	\begin{align*}
		\frac{1}{2}\Xi_{1}^{min}(\varphi_{0})+\frac{1}{2}&\Xi_{2}^{min}(\varphi_{0})\\
		&=\frac{1}{2}\Xi_{1}^{min}(\varphi_{0})+\frac{1}{2}\Xi_{1}^{min}\left(\frac{\pi}{2}-\varphi_{0}\right)\\
		&\le\Xi_{1}^{min}\left(\frac{1}{2}\varphi_{0}+\frac{1}{2}\left(\frac{\pi}{2}-\varphi_{0}\right)\right)\\
		&=\Xi_{1}^{min}\left(\frac{\pi}{4}\right).
	\end{align*}
	Since we can see form \eqref{eq:Xi_arccos} that $\Xi_{1}^{min}\left(\frac{\pi}{4}\right)=-\frac{\pi}{4}$,
	\[
	\Xi_{1}^{min}(\varphi_{0})+\Xi_{2}^{min}(\varphi_{0})\le-\frac{\pi}{2}
	\]
	holds for any $\varphi_{0}\in(0, \frac{\pi}{2})$.
\end{proof}
According to Lemma \ref{lem:decrease} and Lemma \ref{lem:Xi_bounds},
\begin{align*}
	\widehat{\xi}_{1}^{min}(\varphi_{0}, \psi_{0})+\widehat{\xi}_{2}^{min}(\varphi_{0}, \psi_{0})
	&<\Xi_{1}^{min}(\varphi_{0})+\Xi_{2}^{min}(\varphi_{0})\\
	&\le-\frac{\pi}{2},
\end{align*}
that is,
\begin{align*}
	\xi_{1}^{min}\left(\hspace{-0.15mm}t\hspace{-0.15mm}=\hspace{-0.15mm}\frac{1}{\sqrt{2}}, \varphi_{0}, \psi_{0}\right)+\xi_{2}^{min}\left(\hspace{-0.15mm}t\hspace{-0.15mm}=\hspace{-0.15mm}\frac{1}{\sqrt{2}}, \varphi_{0}, \psi_{0}\right)
	<-\frac{\pi}{2}
\end{align*}
holds for any $\varphi_{0}$ and $\psi_{0}$ (i.e. for any $\varphi_{1}$ and $\varphi_{2}$).
However, we cannot conclude that 
\begin{align}
	\label{eq:suff small t}
	\xi_{1}^{min}\left(t, \varphi_{0}, \psi_{0}\right)+\xi_{2}^{min}\left(t, \varphi_{0}, \psi_{0}\right)
	\le-\frac{\pi}{2}
\end{align}
holds for $t\sim\frac{1}{\sqrt{2}}$: it may fail when 
\begin{align*}
	\sup_{\varphi_{0}, \psi_{0}}	\left[\xi_{1}^{min}\left(\hspace{-0.15mm}t\hspace{-0.15mm}=\hspace{-0.15mm}\frac{1}{\sqrt{2}}, \varphi_{0}, \psi_{0}\right)\right.\qquad\qquad\qquad\qquad\\
	\left.+\xi_{2}^{min}\left(\hspace{-0.15mm}t\hspace{-0.15mm}=\hspace{-0.15mm}\frac{1}{\sqrt{2}}, \varphi_{0}, \psi_{0}\right)\right]=-\frac{\pi}{2}.
\end{align*}
On the other hand, because we can observe similarly to Lemma \ref{lem:decrease} that $\xi_{1}^{min}$ is a strictly decreasing function of $\psi_{0}$, 
it is anticipated that \eqref{eq:suff small t} holds for $t\sim\frac{1}{\sqrt{2}}$ and for $\psi_{0}$ sufficiently close to $\frac{\pi}{2}$.
In fact, for $\psi_{0}\in[\frac{\pi}{4}, \frac{\pi}{2})$, we can prove the following proposition.

\begin{proposition}
	\label{prop:xi1_detail_1}
	There exists a constant $C<-\frac{\pi}{2}$ such that 
	\begin{align*}
		\widehat{\xi}_{1}^{min}(\varphi_{0}, \psi_{0})+\widehat{\xi}_{2}^{min}(\varphi_{0}, \psi_{0})
		<C,
	\end{align*}
	i.e.
	\begin{align*}
		\xi_{1}^{min}\left(t=\frac{1}{\sqrt{2}}, \varphi_{0}, \psi_{0}\right)+\xi_{2}^{min}\left(t=\frac{1}{\sqrt{2}}, \varphi_{0}, \psi_{0}\right)
		<C,
	\end{align*}
	holds for all $\psi_{0}\in[\frac{\pi}{4}, \frac{\pi}{2})$ and $\varphi_{0}\in(0, \frac{\pi}{2})$.
\end{proposition}
\begin{proof}
	Because 
	\begin{align*}
		&\widehat{\xi}_{1}^{min}(\varphi_{0}, \psi_{0})+\widehat{\xi}_{2}^{min}(\varphi_{0}, \psi_{0})\\
		&\qquad\qquad=
		\widehat{\xi}_{1}^{min}(\varphi_{0}, \psi_{0})+\widehat{\xi}_{1}^{min}\left(\frac{\pi}{2}-\varphi_{0}, \psi_{0}\right),
	\end{align*}
	we can assume without loss of generality that $0<\varphi_{0}\le\frac{\pi}{4}$.
	Due to Lemma \ref{lem:decrease}, it holds for any $\psi_{0}\in[\frac{\pi}{4}, \frac{\pi}{2})$ that
	\begin{equation}
		\begin{aligned}
			&\widehat{\xi}_{1}^{min}(\varphi_{0}, \psi_{0})
			\le\widehat{\xi}_{1}^{min}\left(\varphi_{0}, \psi_{0}=\frac{\pi}{4}\right),\\
			&\widehat{\xi}_{1}^{min}\left(\frac{\pi}{2}-\varphi_{0}, \psi_{0}\right)
			\le\widehat{\xi}_{1}^{min}\left(\frac{\pi}{2}-\varphi_{0}, \psi_{0}=\frac{\pi}{4}\right).
		\end{aligned}
	\end{equation}
	Let us denote $\widehat{\xi}_{1}^{min}\left(\varphi_{0}, \psi_{0}=\frac{\pi}{4}\right)$ 
	simply by $\widetilde{{\Xi}}_{1}^{min}\left(\varphi_{0}\right)$. 
		In order to investigate $\widetilde{{\Xi}}_{1}^{min}\left(\varphi_{0}\right)$ and $\widetilde{{\Xi}}_{1}^{min}\left(\frac{\pi}{2}-\varphi_{0}\right)$, we have to recall \eqref{eq:xi1_hat}. 
	Similarly to \eqref{eq:Xi1_eq} and \eqref{eq:Xi1=H1} in the proof of Lemma \ref{lem:Xi_bounds}, it results in
	\begin{equation}
		\label{eq:xi1_hat_pi/4}
		\begin{aligned}
			\cos\widetilde{{\Xi}}_{1}^{min}
			=\frac{1}{\sqrt{2}}\cdot\frac{2+\widetilde{H}_{1}\sqrt{2\widetilde{H}_{1}^{2}+2}}{\widetilde{H}_{1}^{2}+2},
		\end{aligned}
	\end{equation}
	where
	\begin{equation}
		\label{eq:h1_hat_pi/4}
		\begin{aligned}
			\widetilde{H}_{1}\left(\varphi_{0}\right)
			&=\widehat{h}_{1}\left(\varphi_{0}, \psi_{0}=\frac{\pi}{4}
			\right)\\
			&=\frac{1}{\sin\varphi_{0}}\left(\frac{1}{\sqrt{2}}-\sqrt{2}\cos\varphi_{0}\right).
		\end{aligned}
	\end{equation}
	Note that  in this case we cannot apply a similar method to the one in Lemma \ref{lem:Xi_bounds} because $\widetilde{\Xi}_{1}^{min}$ does not have a clear form like \eqref{eq:Xi_arccos}.
	Alternatively, we focus on the following monotone relations 
	between $\widetilde{\Xi}_{1}^{min}$, $\widetilde{H}_{1}$, and $\varphi_{0}$ (referring to the proof of Lemma \ref{lem:decrease} may be helpful):
	\begin{align}
		\label{eq:monotones}
		\frac{d\widetilde{\Xi}_{1}^{min}}{d\widetilde{H}_{1}}>0,
		\quad
		\frac{d\widetilde{H}_{1}}{d\varphi_{0}}>0
		\quad
		\left(\mbox{thus}\ \ \frac{d\widetilde{\Xi}_{1}^{min}}{d\varphi_{0}}>0\right).
	\end{align}
	From these relations, it can be seen that our restriction $0<\varphi_{0}\le\frac{\pi}{4}$ is equivalent to the condition $\widetilde{H}_{1}\le1-\sqrt{2}$ since $\widetilde{H}_{1}\left(0\right)=-\infty$ and $\widetilde{H}_{1}\left(\frac{\pi}{4}\right)=1-\sqrt{2}$.
	The claim of the proposition can be shown easily when $\widetilde{H}_{1}\le-1$ (or $0<\varphi_{0}\le\varphi^{*}:=\arccos\frac{2+\sqrt{10}}{6}$, where $\widetilde{H}_{1}(\varphi^{*})=-1$).
	\begin{figure}[h]
		\includegraphics[bb=0.000000 0.000000 702.000000 519.000000, scale=0.315]{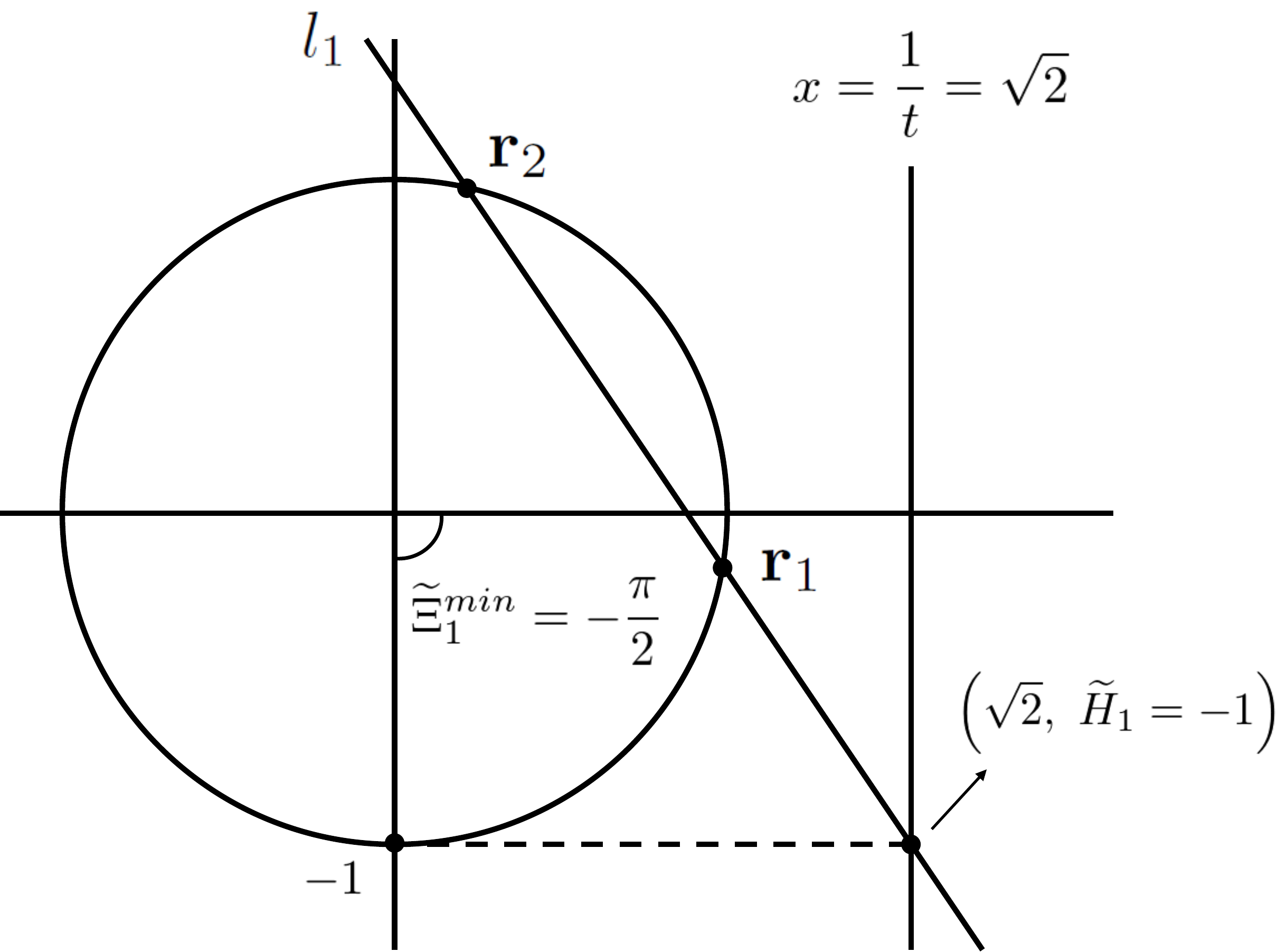}
		\caption{Geometric description of $\widetilde{\Xi}_{1}^{min}$.
			It can be observed that $\widetilde{\Xi}_{1}^{min}=-\frac{\pi}{2}$ when $\widetilde{H}_{1}=-1$.}
		\label{Fig:Xi=-pi/2}
	\end{figure}
	In fact, 
	\begin{align*}
		\widetilde{\Xi}_{1}^{min}(\varphi_{0})\le\widetilde{\Xi}_{1}^{min}\left(\varphi^{*}\right)=-\frac{\pi}{2}
	\end{align*}
	and
	\begin{align*}
		\widetilde{\Xi}_{1}^{min}\left(\frac{\pi}{2}-\varphi_{0}\right)<\widetilde{\Xi}_{1}^{min}\left(\frac{\pi}{2}\right)=-\arccos\frac{2\sqrt{2}+\sqrt{3}}{5}
	\end{align*}
	hold (see FIG. \ref{Fig:Xi=-pi/2} and \eqref{eq:xi1_hat_pi/4}), and thus we can conclude
	\[
	\widetilde{\Xi}_{1}^{min}(\varphi_{0})+\widetilde{\Xi}_{1}^{min}\left(\frac{\pi}{2}-\varphi_{0}\right)<C_{1},
	\]
	where
	\[
	C_{1}=-\frac{\pi}{2}-\arccos\frac{2\sqrt{2}+\sqrt{3}}{5}\left(<-\frac{\pi}{2}\right).
	\]
	When $-1<\widetilde{H}_{1}\le1-\sqrt{2}$ (or $\varphi^{*}<\varphi_{0}\le\frac{\pi}{4}$), we need a bit complicated evaluations.
	It holds similarly to the previous calculations that
	\begin{align*}
		&\widetilde{\Xi}_{1}^{min}(\varphi_{0})\le\widetilde{\Xi}_{1}^{min}\left(\frac{\pi}{4}\right),
		\\
		&\widetilde{\Xi}_{1}^{min}\left(\frac{\pi}{2}-\varphi_{0}\right)<\widetilde{\Xi}_{1}^{min}\left(\frac{\pi}{2}-\varphi^{*}\right).
	\end{align*}
	Since \[\widetilde{\Xi}_{1}^{min}=-\frac{\pi}{4}\iff\widetilde{H}_{1}=0\iff\varphi_{0}=\frac{\pi}{3},
	\]  $\widetilde{\Xi}_{1}^{min}(\frac{\pi}{4})<-\frac{\pi}{4}=\widetilde{\Xi}_{1}^{min}(\frac{\pi}{3})$ holds due to the monotone relations \eqref{eq:monotones}.
	On the other hand, we have
	\begin{align*}
		\cos\varphi^{*}-\cos\frac{\pi}{6}
		&=\frac{2+\sqrt{10}}{6}-\frac{\sqrt{3}}{2}\\
		&=-0.0056...<0,
	\end{align*}
	that is,
	\[
	\varphi^{*}>\frac{\pi}{6}.
	\]
	It follows that $\frac{\pi}{2}-\varphi^{*}<\frac{\pi}{3}$, and thus $\widetilde{\Xi}_{1}^{min}\left(\frac{\pi}{2}-\varphi_{0}\right)<-\frac{\pi}{4}$.
	Therefore, we can conclude also in this case
	\[
	\widetilde{\Xi}_{1}^{min}(\varphi_{0})+\widetilde{\Xi}_{1}^{min}\left(\frac{\pi}{2}-\varphi_{0}\right)<C_{2},
	\]
	where
	\[
	C_{2}=\widetilde{\Xi}_{1}^{min}\left(\frac{\pi}{4}\right)+\widetilde{\Xi}_{1}^{min}\left(\frac{\pi}{2}-\varphi^{*}\right)\left(<-\frac{\pi}{2}\right).
	\]
	Overall, we have obtained
	\[
	\widetilde{\Xi}_{1}^{min}(\varphi_{0}, \psi_{0})+\widetilde{\Xi}_{2}^{min}\left(\varphi_{0}, \psi_{0}\right)<\max\{C_{1}, C_{2}\}\left(<-\frac{\pi}{2}\right)
	\]
	for all $\varphi_{0}\in(0, \frac{\pi}{2})$ and $\psi_{0}\in[\frac{\pi}{4}, \frac{\pi}{2})$.
\end{proof}

By virtue of this proposition, for $t$ sufficiently close to $\frac{1}{\sqrt{2}}$, 
\begin{align*}
	\xi_{1}^{min}\left(t, \varphi_{0}, \psi_{0}\right)+\xi_{2}^{min}\left(t, \varphi_{0}, \psi_{0}\right)
	\le-\frac{\pi}{2}
\end{align*}
follows from the continuity of $\xi_{1}^{min}$ and $\xi_{2}^{min}$ with respect to $t$ when $\frac{\pi}{4}\le\psi_{0}<\frac{\pi}{2}$.
It means that there always exist $\xi_{1}^{\star}\ge\xi_{1}^{min}$ and $\xi_{2}^{\star}\ge\xi_{2}^{min}$ for such $t$ and for any $\varphi_{1}$ and $\varphi_{2}$ satisfying $\xi_{1}^{\star}+\xi_{2}^{\star}=-\frac{\pi}{2}.$
For these $\xi_{1}^{\star}$ and $\xi_{2}^{\star}$, $\mathbf{v}_1=-\mathbf{v}_2$ holds, and thus $\widetilde{\A}_{1}^{t}$ and $\widetilde{\A}_{2}^{t}$ are compatible, i.e. $\A_{D}^{t\vx}$ and $\A_{D}^{t\vy}$ 
are $\state_1$-compatible.

On the other hand, when $0<\psi_{0}<\frac{\pi}{4}$, it may not hold for $t\sim\frac{1}{\sqrt{2}}$ that $\xi_{1}^{min}\left(t, \varphi_{0}, \psi_{0}\right)+\xi_{2}^{min}\left(t, \varphi_{0}, \psi_{0}\right)
\le-\frac{\pi}{2}$, and thus we cannot apply the same argument.
Nevertheless, we can demonstrate that there exist $\xi_1$ and $\xi_2$ 
such that $\widetilde{\A}^{t}_1$ and $\widetilde{\A}^{t}_2$ are compatible even when $0<\psi_{0}<\frac{\pi}{4}$.
To see this, let us assume $0<\psi_{0}<\frac{\pi}{4}$ and apply the necessary and sufficient condition 
for (in)compatibility. 
According to the result proved in \cite{StReHe08,BuSc10,YuLiLiOh10}, $\widetilde{\A}^{t}_1$ and $\widetilde{\A}^{t}_2$ with \eqref{eq:E1_Bloch} and \eqref{eq:E2_Bloch} respectively are 
compatible if and only if
\begin{align}
	\label{eq:compatible iff}
	\begin{aligned}
		&\left( 1- F_1^2 - F_2^2\right)
		\left( 1- \frac{w_1^2}{F_1^2} - \frac{w_2^2}{F_2^2}
		\right)\\
		&\qquad\qquad\qquad\qquad\quad\leq \left( \mathbf{m}_1 \cdot \mathbf{m}_2 
		- w_1 w_2\right)^2
	\end{aligned} 
\end{align}
holds, where 
\begin{align}
	F_1 &:= \frac{1}{2}\hspace{-0.25mm}
	\left(
	\sqrt{(1+w_1)^2 - C_1^2}
	+ \sqrt{(1-w_1)^2 - C_1^2}
	\right),
	\\
	F_2 &:= \frac{1}{2}\hspace{-0.25mm}
	\left(
	\sqrt{(1+w_2)^2 - C_2^2}
	+ \sqrt{(1-w_2)^2 - C_2^2}
	\right). 
\end{align}
For $\xi_1^{min}$ and $\xi_2^{min}$, 
since it holds that 
\begin{align}
	1-w_1(\xi_1^{min}) &= C_1 (\xi_1^{min}),\\
	1-w_2(\xi_2^{min}) &= C_2(\xi_2^{min}), 
\end{align}
they become 
\begin{align}
	F_1 = \sqrt{w_1(\xi_1^{min})}, \quad
	F_2= \sqrt{w_2(\xi_2^{min}) }. 
\end{align}
Therefore, \eqref{eq:compatible iff} can be rewritten as
\begin{widetext}
	\begin{equation}
		\label{eq:compatible iff for min}
		\begin{aligned}
			\left[
			(1- \sin (\xi_1^{min}+ \xi_2^{min}))
			w_1(\xi_1^{min})w_2(\xi_2^{min})
			-
			(1+\sin(\xi_1^{min} + \xi_2^{min}))
			(1- w_1(\xi_1^{min}) - w_2(\xi_2^{min}))
			\right]\qquad\qquad\qquad\\
			\left[
			(1-w_1(\xi_1^{min}))(1-w_2(\xi_2^{min}))(1- \sin (\xi_1^{min}+ \xi_2^{min}))
			\right]
			\ge 0. 
		\end{aligned}
	\end{equation}
\end{widetext}
If $1-\sin(\xi_1^{min} + \xi_2^{min})=0$, then \eqref{eq:compatible iff for min} holds, that is, $\widetilde{\A}^{t}_1$ and $\widetilde{\A}^{t}_2$ for $\xi_1^{min}$ and $\xi_2^{min}$ respectively are compatible.
Therefore, we hereafter assume $1-\sin(\xi_1^{min} + \xi_2^{min})>0$, and rewrite \eqref{eq:compatible iff for min} as (note that $0<w_1(\xi_1^{min})<1$, $0<w_2(\xi_2^{min})<1$)
\begin{align}
	\begin{aligned}
		(1+\sin(\xi_1^{min} + \xi_2^{min}))
		(1- w_1(\xi_1^{min}) - w_2(\xi_2^{min}))\ \\
		\leq (1- \sin (\xi_1^{min}+ \xi_2^{min}))
		w_1(\xi_1^{min})w_2(\xi_2^{min}). 
	\end{aligned}
\end{align}
In other words, 
$\widetilde{\A}^{t}_1$ and 
$\widetilde{\A}^{t}_2$ with respect to 
$\xi^{min}_1$ and $\xi^{min}_2$ are incompatible 
if and only if
\begin{align}
	\label{eq:incompatible iff for min}
	\begin{aligned}
		(1+\sin(\xi_1^{min} + \xi_2^{min}))
		(1- w_1(\xi_1^{min}) - w_2(\xi_2^{min}))\ \\
		> (1- \sin (\xi_1^{min}+ \xi_2^{min}))
		w_1(\xi_1^{min})w_2(\xi_2^{min})
	\end{aligned}
\end{align}
holds.
In order to investigate whether \eqref{eq:incompatible iff for min} holds, it is helpful to introduce a function $Z$ defined as 
\begin{widetext}
	\begin{equation}
		\label{eq:def of Z}
		\begin{aligned}
			Z(t, \varphi_{0}, \psi_{0})
			:= \left[1+\sin(\xi_1^{min}(t, \varphi_{0}, \psi_{0})+\xi_2^{min}(t, \varphi_{0}, \psi_{0}))\right]
			\left[1+ w_1(\xi_1^{min}(t, \varphi_{0}, \psi_{0})) + w_2(\xi_2^{min}(t, \varphi_{0}, \psi_{0}))\right]
			\qquad\quad\\
			- \left[1- \sin (\xi_1^{min}(t, \varphi_{0}, \psi_{0})+\xi_2^{min}(t, \varphi_{0}, \psi_{0}))\right]
			w_1(\xi_1^{min}(t, \varphi_{0}, \psi_{0}))
			w_2(\xi_2^{min}(t, \varphi_{0}, \psi_{0})). 
		\end{aligned}
	\end{equation}
\end{widetext}
Because
\begin{align*}
	(1+\sin(\xi_1^{min} + \xi_2^{min}))
	(1- w_1(\xi_1^{min}) - w_2(\xi_2^{min}))\qquad\\
	<(1+\sin(\xi_1^{min} + \xi_2^{min}))
	(1+ w_1(\xi_1^{min}) + w_2(\xi_2^{min})),
\end{align*}
\begin{eqnarray}
	\label{eq:Z>0}
	Z(t, \varphi_{0}, \psi_{0}) >0
\end{eqnarray}
holds if $\widetilde{\A}^{t}_1$ and 
$\widetilde{\A}^{t}_2$ with respect to 
$\xi^{min}_1$ and $\xi^{min}_2$ are incompatible.
Let us focus on the case when $t=\frac{1}{\sqrt{2}}$ (i.e. $\xi_{1}^{min}=\widehat{\xi}_{1}^{min}$).
If a pair $(\varphi_{0}, \psi_{0})$ satisfies $\widehat{\xi}_{1}^{min}(\varphi_{0}, \psi_{0})\le-\frac{\pi}{2}$ or $\widehat{\xi}_{2}^{min}(\varphi_{0}, \psi_{0})\le-\frac{\pi}{2}$, then 
\begin{align*}
	\widehat{\xi}_{1}^{min}(\varphi_{0}, \psi_{0})+\widehat{\xi}_{2}^{min}(\varphi_{0}, \psi_{0})
	<C
\end{align*}
with
\begin{align*}
	C&=-\frac{\pi}{2}+
	\lim_{\substack{\varphi_{0}\rightarrow \frac{\pi}{2}-0\\ 
			\psi_{0}\rightarrow+0}}
	\widehat{\xi}_{1}^{min}\left(\varphi_{0}, \psi_{0}\right)\\
	&=-\frac{\pi}{2}-\arccos\left(\frac{2\sqrt{2}}{3}\right)\\
	&<-\frac{\pi}{2}
\end{align*}
holds due to similar monotone relations to \eqref{eq:monotones} between $\varphi_{0}, \psi_{0},$ and $\widehat{\xi}_{1}^{min}$ (remember that $\widehat{\xi}_{2}^{min}(\varphi_{0}, \psi_{0})=\widehat{\xi}_{1}^{min}\left(\frac{\pi}{2}-\varphi_{0}, \psi_{0}\right)$).
Therefore, in this case, we can apply the same argument as Proposition \ref{prop:xi1_detail_1}, which results in the compatibility of $\widetilde{\A}^{t}_1$ and 
$\widetilde{\A}^{t}_2$ for $t\sim\frac{1}{\sqrt{2}}$.
On the other hand, let us examine the case when  $(\varphi_{0}, \psi_{0})$ satisfies $\psi_{0}\in(0, \frac{\pi}{4})$, and $\widehat{\xi}_{1}^{min}(\varphi_{0}, \psi_{0})>-\frac{\pi}{2}$ and $\widehat{\xi}_{2}^{min}(\varphi_{0}, \psi_{0})>-\frac{\pi}{2}$.
Because $\psi_{0}\in(0, \frac{\pi}{4})$, we obtain for general $t$ (see \eqref{eq:x_1explicit})
\begin{align}
	\begin{aligned}
		w_1(\xi^{min}_1) 
		&> - \frac{t}{\sqrt{2}} \frac{\sin \xi_1^{min}}{
			\sin (\varphi_0- \xi_1^{min})}\\
		&\geq \frac{t}{\sqrt{2}} (-\sin \xi_1^{min})
		. 
	\end{aligned}
\end{align}
For $t=\frac{1}{\sqrt{2}}$, since 
\[
-\frac{\pi}{2}<\widehat{\xi}_{1}^{min}(\varphi_{0}, \psi_{0})<\lim_{\substack{\varphi_{0}\rightarrow \frac{\pi}{2}-0\\ 
		\psi_{0}\rightarrow+0}}
\widehat{\xi}_{1}^{min}\left(\varphi_{0}, \psi_{0}\right),
\] 
it gives a bound 
\begin{eqnarray}
	w_1(\xi_1^{min}) >  \frac{1}{2} \sin \widehat{\xi}_0, 
\end{eqnarray}
where we define
\begin{align*}
	\widehat{\xi}_0 &= -\lim_{\substack{\varphi_{0}\rightarrow \frac{\pi}{2}-0\\ 
			\psi_{0}\rightarrow+0}}
	\widehat{\xi}_{1}^{min}\left(\varphi_{0}, \psi_{0}\right)\\
	&=\arccos \left(\frac{2\sqrt{2}}{3}\right).
\end{align*}
Let $\varepsilon$ be a positive constant satisfying $\varepsilon<\frac{1}{16} (\sin\widehat{\xi}_0)^{2}$.
Due to the continuity of sine, there exists a positive constant $\delta$ such that $\sin x\in(-1, -1+\varepsilon)$ whenever $x\in\left(-\frac{\pi}{2}-\delta, -\frac{\pi}{2}\right)$.
If $(\varphi_{0}, \psi_{0})$ satisfies $\widehat{\xi}_{1}^{min}(\varphi_{0}, \psi_{0})+\widehat{\xi}_{2}^{min}(\varphi_{0}, \psi_{0})\le-\frac{\pi}{2}-\delta$, then it again leads to the same argument as Proposition \ref{prop:xi1_detail_1}, and we can see that $\widetilde{\A}^{t}_1$ and $\widetilde{\A}^{t}_2$ for this $(\varphi_{0}, \psi_{0})$ are compatible.
Conversely, if $(\varphi_{0}, \psi_{0})$ satisfies $-\frac{\pi}{2}-\delta<\widehat{\xi}_{1}^{min}(\varphi_{0}, \psi_{0})+\widehat{\xi}_{2}^{min}(\varphi_{0}, \psi_{0})<-\frac{\pi}{2}$ (remember Lemma \ref{lem:Xi_bounds}), then 
\[
-1<\sin(
\widehat{\xi}_{1}^{min}(\varphi_{0}, \psi_{0})+\widehat{\xi}_{2}^{min}(\varphi_{0}, \psi_{0})
)<-1+\varepsilon
\]
follows from the definition of $\delta$.
Therefore, by virtue of \eqref{eq:def of Z}, we have
\begin{widetext}
	\begin{align*}
		Z\left(t=\frac{1}{\sqrt{2}}, \varphi_{0}, \psi_{0}\right)
		&=\left[1+\sin(\widehat{\xi}_1^{min}(\varphi_{0}, \psi_{0})+\widehat{\xi}_2^{min}(\varphi_{0}, \psi_{0}))\right]
		\left[1+ w_1(\widehat{\xi}_1^{min}(\varphi_{0}, \psi_{0})) + w_2(\widehat{\xi}_2^{min}(\varphi_{0}, \psi_{0}))\right]
		\\
		&\qquad\qquad\qquad\qquad
		- \left[1- \sin (\widehat{\xi}_1^{min}(\varphi_{0}, \psi_{0})+\widehat{\xi}_2^{min}(\varphi_{0}, \psi_{0}))\right]
		w_1(\widehat{\xi}_1^{min}(\varphi_{0}, \psi_{0}))
		w_2(\widehat{\xi}_2^{min}(\varphi_{0}, \psi_{0}))
		\\
		&<\varepsilon\left[1+ w_1(\widehat{\xi}_1^{min}(\varphi_{0}, \psi_{0})) + w_2(\widehat{\xi}_2^{min}(\varphi_{0}, \psi_{0}))\right]-(2-\varepsilon)w_1(\widehat{\xi}_1^{min}(\varphi_{0}, \psi_{0}))
		w_2(\widehat{\xi}_2^{min}(\varphi_{0}, \psi_{0}))
		\\
		&=\varepsilon\left[1+w_1(\widehat{\xi}_1^{min}(\varphi_{0}, \psi_{0}))\right]\left[1+w_2(\widehat{\xi}_1^{min}(\varphi_{0}, \psi_{0}))\right]-2w_1(\widehat{\xi}_1^{min}(\varphi_{0}, \psi_{0}))
		w_2(\widehat{\xi}_2^{min}(\varphi_{0}, \psi_{0}))
		\\
		&<4\varepsilon-2w_1(\widehat{\xi}_1^{min}(\varphi_{0}, \psi_{0}))
		w_2(\widehat{\xi}_2^{min}(\varphi_{0}, \psi_{0})).
	\end{align*}
\end{widetext}
Because 
\begin{align*}
	4\varepsilon-2w_1(\widehat{\xi}_1^{min}(\varphi_{0}, \psi_{0}))
	w_2(\widehat{\xi}_2^{min}(\varphi_{0}, \psi_{0}))\qquad\quad
	\\<\frac{1}{4} (\sin\widehat{\xi}_0)^{2}-\frac{1}{2} (\sin\widehat{\xi}_0)^{2}=-\frac{1}{4} (\sin\widehat{\xi}_0)^{2},
\end{align*}
it holds that
\[
Z\left(t=\frac{1}{\sqrt{2}}, \varphi_{0}, \psi_{0}\right)
<-\frac{1}{4} (\sin\widehat{\xi}_0)^{2}<0.
\]
Therefore, for $t\sim\frac{1}{\sqrt{2}}$, $Z\left(t, \varphi_{0}, \psi_{0}\right)\le0$ holds, that is, $\widetilde{\A}_1^{t}$ and 
$\widetilde{\A}_2^{t}$ with respect to 
$\xi^{min}_1$ and $\xi^{min}_2$ are compatible.
Overall, we have demonstrated that when $t\sim\frac{1}{\sqrt{2}}$, there exist compatible observables $\widetilde{\A}_1^{t}$ and 
$\widetilde{\A}_2^{t}$ for any line $\state_{1}\subset\state_{D}$ such that they agree with $\A^{t\vx}_{D}$ and $\A^{t\vy}_{D}$ on $\state_{1}$ respectively.
That is, when $t\sim\frac{1}{\sqrt{2}}$, $\A^{t\vx}_{D}$ and $\A^{t\vy}_{D}$ are $\state_{1}$-compatible for any line $\state_{1}\subset\state_{D}$.
Therefore, we can conclude that $\chi_{incomp}(\A^{t\vx}_{D},\A^{t\vy}_{D})=3$ for $t\sim\frac{1}{\sqrt{2}}$, and thus the set $M$ in \eqref{eq: sets of xi=2,3} is nonempty.

\subsection*{Proof of Proposition \ref{prop:qubit-threshold}
	: Part 2}
In this part, we shall show that 
\[
t_{0}':=\inf L=\sup M\in M,
\]
where $L$ and $M$ are defined in \eqref{eq: sets of xi=2,3}.
In order to prove this, we will see that if $t\in L$, then $t-\delta\in L$ for sufficiently 
small $\delta>0$, that is, $t_{0}'\notin L$.

Let us focus again on a system described by 
a two-dimensional disk
state space $\mathcal{S}_D$. 
It is useful to identify this system with the system of a quantum bit with real coefficients by replacing 
$\{\sigma_1, \sigma_2\}$ with 
$\{\sigma_3, \sigma_1\}$.
Then, defining $\mathcal{E}_{D}$ as the set of all effects on $\mathcal{S}_D$, we can see that any $E\in\mathcal{E}_{D}$  can be expressed as a real-coefficient positive matrix smaller than $\id$.
We also define $O_{D}(2)
\subset \mathcal{E}_{D}\times \mathcal{E}_{D}$ as the set of all binary observables  on $\state_{D}$, 
which is isomorphic naturally to $\mathcal{E}_{D}$
since a binary observable $\A$ is completely specified by its effect $\A(+) \in \mathcal{E}_{D}$. 
With introducing a topology (e.g. norm topology) on $\mathcal{E}_{D}$, it also can be observed that $O_{D}(2)$ is homeomorphic to $\mathcal{E}_{D}$.
Note that because the system is described by finite dimensional matrices, any (natural) topology (norm topology, weak topology, etc.) coincides with each other.  
For a pair of states $\{\varrho^{\mathbf{r}_{1}}, \varrho^{\mathbf{r}_2}\}$ in $\state_{D}$, and a binary observable $\A
\in O_{D}(2)$, we define a set of observables $C(\A: \varrho^{\mathbf{r}_{1}}, 
\varrho^{\mathbf{r}_{2}})$ as the set of all binary observables $\widetilde{\A}\in O_{D}(2)$ such that 
\begin{align*}
	&\tr{\varrho^{\mathbf{r}_{1}}\widetilde{\A}(\pm)}=\tr{\varrho^{\mathbf{r}_{1}}\A(\pm)},\\ &\tr{\varrho^{\mathbf{r}_{2}}\widetilde{\A}(\pm)}=\tr{\varrho^{\mathbf{r}_{2}}\A(\pm)}.
\end{align*}
It can be confirmed easily that  
$C(\A: \varrho^{\mathbf{r}_{1}}, 
\varrho^{\mathbf{r}_{2}})$ is closed in $O_{D}(2)\simeq \mathcal{E}_{D}$. 
Let us denote by $O_{D}(4)$ the set of all observables 
with four outcomes, which is a compact (i.e. bounded and closed) subset of 
$\mathcal{E}_{D}^4$.  
For each $\M=\{\M(x,y)\}\in O_{D}(4)$, 
we can introduce a pair of binary observables 
by 
\[
\pi_1(\M)= \left\{\sum_{y}\M(x,y)\right\}_x,\  
\pi_2(\M)=\left\{\sum_x \M(x,y)\right\}_y.
\]
Since $\pi_j\colon O_{D}(4)\to O_{D}(2)$ is continuous, 
the set of all compatible binary observables 
denoted by
\begin{align*}
	JM(2,2):=\{(\pi_1(\M), \pi_2(\M))\mid\M \in O_{D}(4)\}
\end{align*}
is compact in $O_{D}(2)\times O_{D}(2)
\simeq \mathcal{E}_{D} \times \mathcal{E}_D$ as well. 
As we have seen in the previous part, $\chi_{incomp}(\A^{t\vx}, \A^{t\vy})=2$ (i.e. $\chi_{incomp}(\A_{D}^{t\vx}, \A_{D}^{t\vy})=2$) if and only if 
there exists a pair of vectors
$\mathbf{r}_{1}, \mathbf{r}_{2} \in \partial D$ such that 
\begin{align*}
	\left(C(\A_{D}^{t\vx}: \varrho^{\mathbf{r}_{1}}, \varrho^{\mathbf{r}_{2}})
	\times C(\A^{t\vy}_{D}: \varrho^{\mathbf{r}_{1}}, 
	\varrho^{\mathbf{r}_{2}}) \right)
	\cap JM(2:2) =\emptyset. 
\end{align*} 

\par
Let us examine concrete representations of the 
sets. 
Each effect $E\in\mathcal{E}_{D}$ is written as 
$E=\frac{1}{2}(e_0 \id + \mathbf{e}\cdot\sigma)=\frac{1}{2}(e_0 \id + e_1 \sigma_1 + e_2\sigma_2)$ 
with $(e_0, \mathbf{e})=(e_0, e_1, e_2) 
\in \real^3$ satisfying $0 \leq  
e_0 \pm |\mathbf{e}| \leq 2$.

If we consider another effect $F=\frac{1}{2}(f_0 \id + \mathbf{f}\cdot\sigma)$, the operator norm of $E-F$ 
is calculated as 
\begin{align}
	\label{eq:norm for qubit effect}
	\Vert E-F\Vert =\frac{1}{2}\left( |e_0 - f_0 | 
	+ |\mathbf{e} - \mathbf{f}|\right). 
\end{align}
We may employ this norm to define 
a topology on $\mathcal{E}_{D}$ and $O_{D}(2)
\simeq \mathcal{E}_{D}$. 
On the other hand, each state in $\state_{D}$ is parameterized as 
$\varrho^{\mathbf{r}_{1}} = \frac{1}{2}(\id + x_1 \sigma_1 + 
y_1\sigma_2)$, where $\mathbf{r}_{1}=(x_1, y_1)$ 
satisfies $|\mathbf{r}_{1}|\leq 1$. 
For an effect $E$ and a state $\varrho^{\mathbf{r}_{1}}$, we have 
$\mbox{tr}[\varrho^{\mathbf{r}_{1}}E]
=\frac{1}{2}(e_0 + \mathbf{r}_{1}\cdot \mathbf{e})
$. 
In particular, when considering $\A^{t\vx}(\pm)=\frac{1}{2}(\id \pm 
t \sigma_1)$, 
a binary observable $\C$ determined by the effect $\C(+)=\frac{1}{2}(c_0 \id  + \mathbf{c}\cdot 
\mathbf{\sigma})=\frac{1}{2}(c_0 \id  + c_{1}\sigma_1+c_{2}\sigma_{2})$ satisfies 
$\C \in 
C(\A^{t\vx}: \varrho^{\mathbf{r}_{1}}, \varrho^{\mathbf{r}_{2}})$ if and only if 
\begin{align*}
	&\tr{\varrho^{\mathbf{r}_{1}}\A^{t\vx}(+)}=\tr{\varrho^{\mathbf{r}_{1}}\C(+)},\\ &\tr{\varrho^{\mathbf{r}_{2}}\A^{t\vx}(+)}=\tr{\varrho^{\mathbf{r}_{2}}\C(+)},
\end{align*}
i.e.
\begin{align*}
	1 +t x_1 &=c_0 + \mathbf{r}_{1}\cdot \mathbf{c}
	= c_0 + x_1 c_1+y_1 c_2, \\
	1 + t x_2 &=c_0 + \mathbf{r}_{2}\cdot \mathbf{c}
	= c_0 + x_2 c_1 +y_2 c_2. 
\end{align*}
hold, where we set $\mathbf{r}_{2}=(x_{2}, y_{2})$.
The set of their solutions for $(c_0, \mathbf{c})$ is represented as 
\begin{align*}
	(c_0, \mathbf{c})=(1, t, 0)+\lambda'
	\left(
	-\frac{x_{1}y_{2}-y_{1}x_{2}}{x_{1}-x_{2}},\  
	-\frac{y_{1}-y_{2}}{x_{1}-x_{2}},\ 
	1
	\right)
\end{align*}
with $\lambda'\in\real$.
Let us define a vector $\mathbf{n}\in\real^{2}$ such that 
\[
(1, \mathbf{n})\cdot(1, \mathbf{r}_{1})=(1, \mathbf{n})\cdot(1, \mathbf{r}_{2})=0
\]
(i.e. $\mathbf{n}\cdot\mathbf{r}_{1}=\mathbf{n}\cdot\mathbf{r}_{2}=-1$).
It is easy to see that 
\[
\left(
-\frac{x_{1}y_{2}-y_{1}x_{2}}{x_{1}-x_{2}},\  
-\frac{y_{1}-y_{2}}{x_{1}-x_{2}},\ 
1
\right)
\propto
(1, \mathbf{n}),
\]
and thus the set of solutions can be rewritten as
\begin{align}
	\label{eq:solution for C}
	(c_0, \mathbf{c})=(1, t, 0)+\lambda(1, \mathbf{n})
\end{align}
with $\lambda\in\real$.
Note that because we are interested in the case when $\A_{D}^{t\vx}$ and $\A_{D}^{t\vy}$ are $\state_{1}$-incompatible, we do not consider the case when $\mathbf{r}_{1}$ and $\mathbf{r}_{2}$ are parallel or when $x_{1}=x_{2}$ corresponding to $\psi_{0}=\frac{\pi}{2}$ or $\varphi_{0}=0$ in Part 1 respectively.
Therefore, the vector $\mathbf{n}=(n_{x}, n_{y})$ can be defined successfully, and it is easy to verify that $|\mathbf{n}|=\sqrt{n_{x}^{2}+n_{y}^{2}}>1$.
\begin{figure}[h]
	\includegraphics[bb=0.000000 0.000000 539.000000 496.000000, scale=0.3]{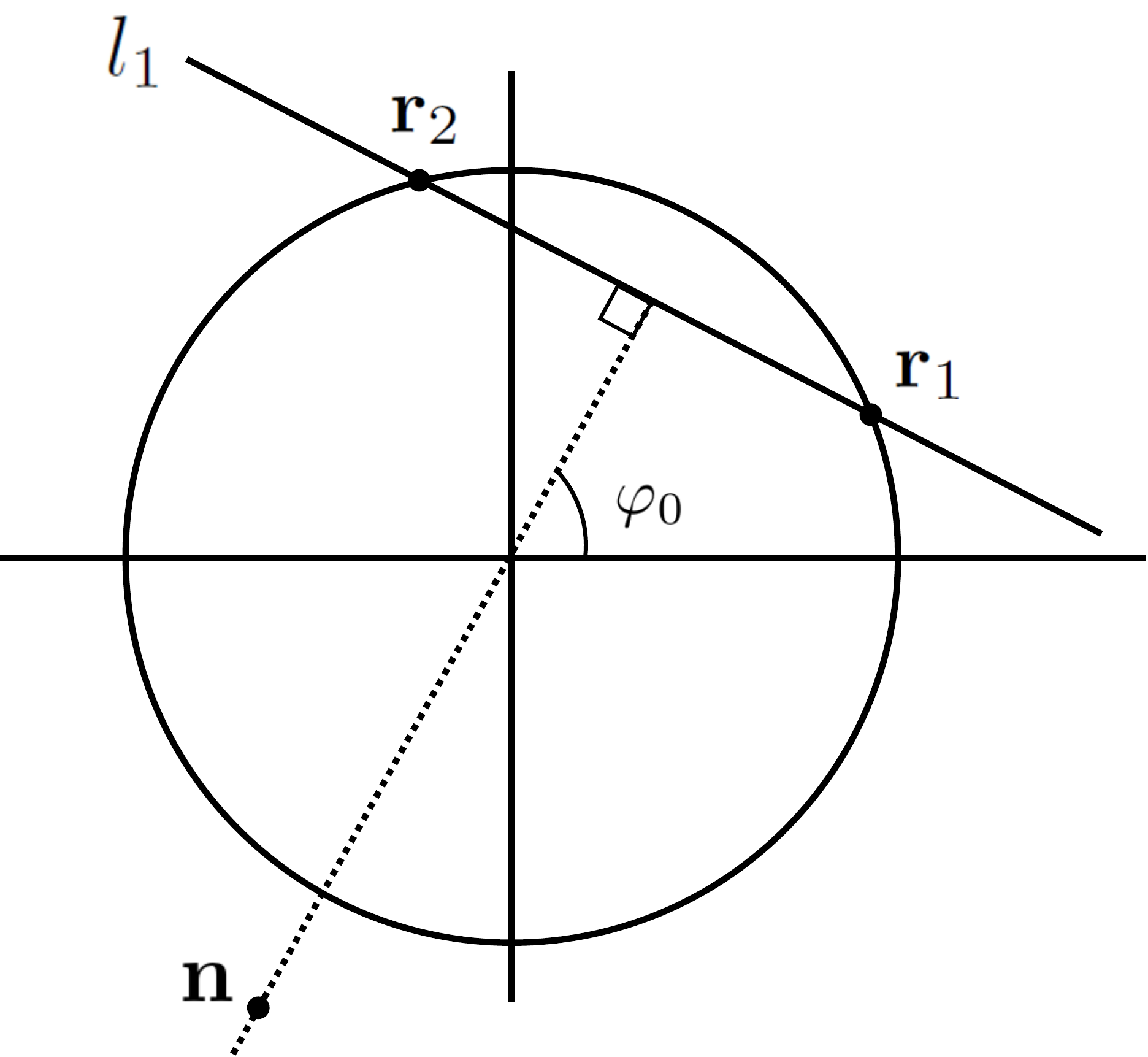}
	\caption{Geometric description of $\mathbf{n}$: we can observe that it lies in the third quadrant.}
	\label{Fig:vector n}
\end{figure}
Moreover, because $\varphi_{0}$ is supposed to be $0<\varphi_{0}<\frac{\pi}{2}$ as shown in Part 1, we can assume without loss of generality that its components $n_{x}$ and $n_{y}$ are negative (see FIG. \ref{Fig:vector n}).
In order for $\C$ to be an element of $C(\A^{t\vx}: \varrho^{\mathbf{r}_{1}}, \varrho^{\mathbf{r}_{2}})$, \eqref{eq:solution for C} should also satisfy
\[
0\leq 1+\lambda \pm |(t, 0) + \lambda\mathbf{n}|
\leq 2,
\]
i.e.
\begin{align*}
	1+\lambda - |(t, 0) + \lambda\mathbf{n}| \ge 0,\quad
	1+\lambda + |(t, 0) + \lambda\mathbf{n}| \le 2.
\end{align*}
It can be reduced to 
\begin{align}
	\label{eq:lambda}
	\lambda_{1}^{t}\le \lambda \le \lambda_{2}^{t}
\end{align}
with
\begin{equation}
	\label{eq:lambda2}
	\begin{aligned}
		&\lambda_{1}^{t}=\frac{1-n_{x}t-\sqrt{(1-n_{x}t)^{2}+(|\mathbf{n}|^{2}-1)(1-t^2)}}{|\mathbf{n}|^{2}-1},\\
		&\lambda_{2}^{t}=\\
		&\min\left\{1,\ \frac{-1-n_{x}t+\sqrt{(1+n_{x}t)^{2}+(|\mathbf{n}|^{2}-1)(1-t^2)}}{|\mathbf{n}|^{2}-1}\right\},
	\end{aligned}
\end{equation}
where we used $|\mathbf{n}|>1$ and $n_{x}<0$ (see FIG. \ref{Fig:lambda}).
Overall, $C(\A^{t\vx}_{D}: \varrho^{\mathbf{r}_{1}}, \varrho^{\mathbf{r}_{2}})$
is isomorphic to the set parameterized as 
\begin{equation}
	\label{eq:explicit C}
	\begin{aligned}
		\{(1, t, 0) + \lambda(1, \mathbf{n})\mid\lambda_{1}^{t}\le \lambda \le \lambda_{2}^{t}\},
	\end{aligned}
\end{equation}
where $\lambda_{1}^{t}$ and $\lambda_{2}^{t}$ are shown in \eqref{eq:lambda2}.
Remark that the same argument can be applied for $C(\A^{t\vy}: \varrho^{\mathbf{r}_{1}}, \varrho^{\mathbf{r}_{2}})$.
\begin{figure}[h]
	\includegraphics[bb=0.000000 0.000000 876.000000 499.000000, scale=0.265]{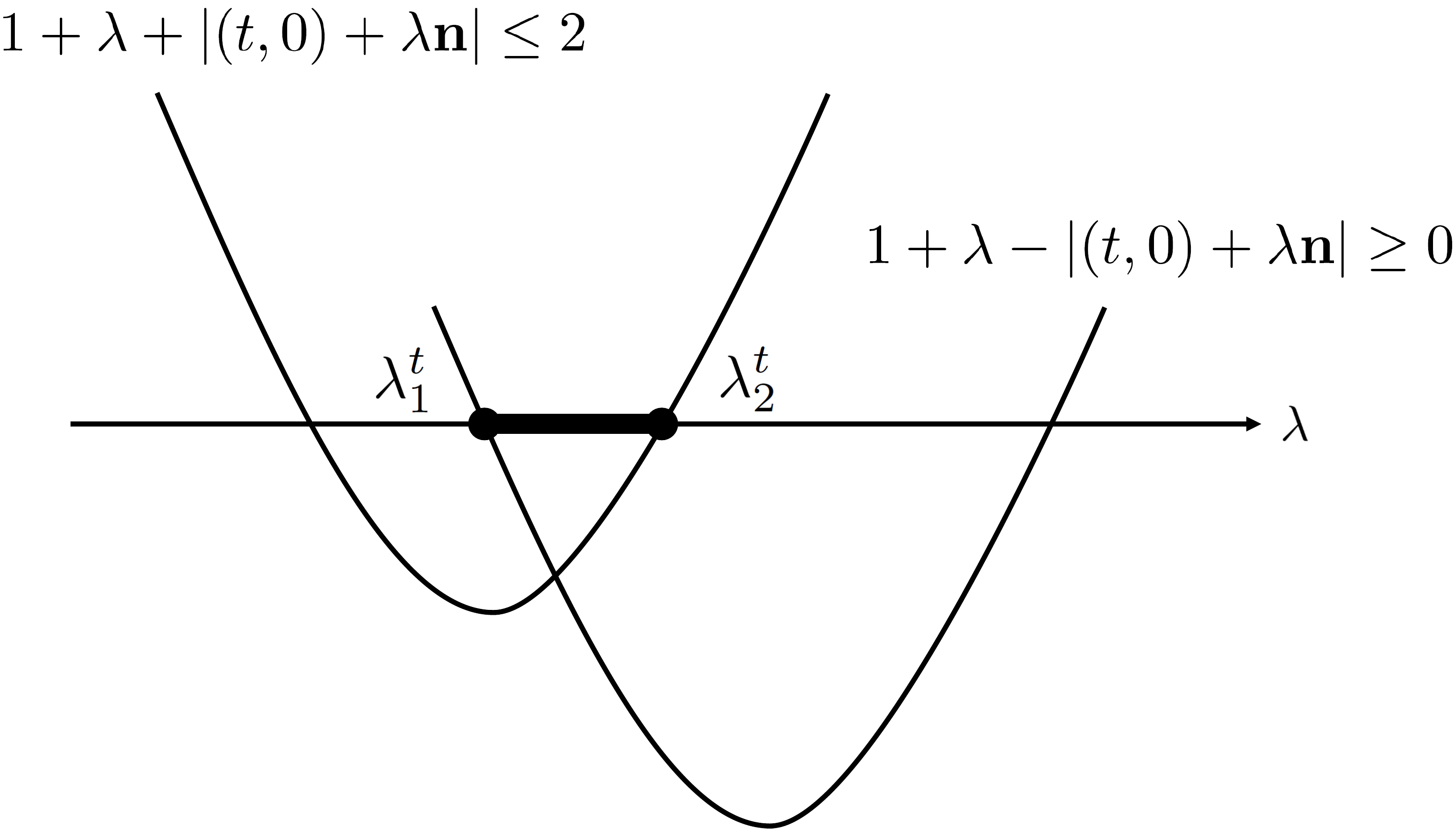}
	\caption{Solutions for $\lambda$.}
	\label{Fig:lambda}
\end{figure}

We shall now prove $t_{0}'=\inf L\notin L$.
Suppose that $t\in L$, i.e. $\chi_{incomp}(\A^{t\vx}_{D}, \A^{t\vy}_{D})=2$.
It follows that there exist $\mathbf{r}_{1}$ and $\mathbf{r}_{2}$ in $\partial D$ such that
\begin{align*}
	\left(C(\A_{D}^{t\vx}: \varrho^{\mathbf{r}_{1}}, \varrho^{\mathbf{r}_{2}})
	\times C(\A^{t\vy}_{D}: \varrho^{\mathbf{r}_{1}}, 
	\varrho^{\mathbf{r}_{2}}) \right)
	\cap JM(2:2) =\emptyset. 
\end{align*} 
Denoting $C(\A_{D}^{t\vx}: \varrho^{\mathbf{r}_{1}}, \varrho^{\mathbf{r}_{2}})$ and $C(\A_{D}^{t\vy}: \varrho^{\mathbf{r}_{1}}, \varrho^{\mathbf{r}_{2}})$ simply by $X^{t}$ and $Y^{t}$ respectively, we can rewrite it as
\[
X^{t}\times Y^{t} \cap JM(2:2) =\emptyset.
\]
We need the following lemma.
\begin{lemma}
	\label{lem: metric}
	Let $\delta>0$.
	There exists $\Delta>0$ such that 
	for all $\tau\in[0, \Delta]$ and for all $\C\in X^{t-\tau}$, there exists $\A\in X^{t}$ satisfying 
	\[
	d(\C, \A):=\|\C(+)-\A(+)\|<\delta
	\]
	where $d$ is a metric on $O_{D}(2)$ defined through the operator norm $\|\cdot\|$ on $\mathcal{E}_{D}\simeq O_{D}(2)$.
\end{lemma}

\begin{proof}
	By its definition, $X^{t}$ is a convex set of $O_{D}(2)$, and thus  for all $\E\in O_{D}(2)$ we can define successfully the distance between $\E$ and $X^{t}$: 
	\[
	d(\E, X^{t})=\min_{\F\in X^{t}}d(\E, \F).
	\]
	In particular, for $\E'\in X^{t-\Delta'}\subset O_{D}(2)$ with $\Delta'>0$ and $\E'(+)=\frac{1}{2}(e'_{0}\id+\mathbf{e}'\cdot\sigma)$, it becomes
	\begin{align}
		\label{eq:d(E', X^t)}
		\begin{aligned}
			d(\E', X^{t})
			&=\min_{\F\in X^{t}}d(\E', \F)\\
			&=\min_{\F\in X^{t}}\frac{1}{2}\left( |e'_0 - f_0 | 
			+ |\mathbf{e} - \mathbf{f}|\right),
		\end{aligned}
	\end{align}
	where $\F(+)=\frac{1}{2}(f_{0}\id+\mathbf{f}\cdot\sigma)$ (see \eqref{eq:norm for qubit effect}).
	Since, in terms of \eqref{eq:explicit C}, $\E'\in X^{t-\Delta'}$ and $\F\in X^{t}$ imply 
	\[
	(e_{0}', \mathbf{e}')=(1, t-\Delta', 0) + \lambda'(1, \mathbf{n})
	\]
	with $\lambda_{1}^{t-\Delta'}\le \lambda' \le \lambda_{2}^{t-\Delta'}$ and 
	\[
	(f_{0}, \mathbf{f})=(1, t, 0) + \lambda(1, \mathbf{n})
	\]
	with $\lambda_{1}^{t}\le \lambda \le \lambda_{2}^{t}$ respectively, \eqref{eq:d(E', X^t)} can be rewritten as
	\begin{align*}
		\begin{aligned}
			2d&(\E', X^{t})\\
			&=\min_{\lambda\in[\lambda_{1}^{t}, \lambda_{2}^{t}]}
			\left(
			|\lambda'-\lambda|+|(-\Delta', 0)+(\lambda'-\lambda)\mathbf{n}|
			\right).
		\end{aligned}
	\end{align*}
	It follows that 
	\begin{align}
		\label{eq:d(E', X^t) bound}
		2d(\E', X^{t})
		\le
		\Delta'+\min_{\lambda\in[\lambda_{1}^{t}, \lambda_{2}^{t}]}|\lambda'-\lambda|(1+|\mathbf{n}|).
	\end{align}
	Let us evaluate its right hand side.
	It is easy to see that 
	\[
	\min_{\lambda\in[\lambda_{1}^{t}, \lambda_{2}^{t}]}|\lambda'-\lambda|
	=\left\{
	\begin{aligned}
		\lambda_{1}^{t}&-\lambda' & \quad&(\lambda'<\lambda_{1}^{t})\\
		&0 & \quad&(\lambda_{1}^{t}\le\lambda'\le\lambda_{2}^{t})\\
		\lambda'&-\lambda_{2}^{t}
		& \quad &(\lambda'>\lambda_{2}^{t})
	\end{aligned}
	\right.
	.
	\]
	Suppose that $\lambda'<\lambda_{1}^{t}$ holds, for example.
	In this case, because $\lambda_{1}^{t-\Delta'}\le\lambda'$, we can obtain
	\[
	\lambda_{1}^{t}-\lambda'
	\le
	\lambda_{1}^{t}-\lambda_{1}^{t-\Delta'}.
	\]
	In a similar way, it can be demonstrated that
	\begin{align*}
		\sup_{\lambda'\in[\lambda_{1}^{t-\Delta'}, \lambda_{2}^{t-\Delta'}]}&\min_{\lambda\in[\lambda_{1}^{t}, \lambda_{2}^{t}]}|\lambda'-\lambda|\\
		=&\max\left\{
		\lambda_{1}^{t}-\lambda_{1}^{t-\Delta'},\ 0,\ 
		\lambda_{2}^{t-\Delta'}-\lambda_{2}^{t}
		\right\}.
	\end{align*}
	By virtue of \eqref{eq:lambda2}, the right hand side converges to 0 as $\Delta'\rightarrow0$, and thus we can see from \eqref{eq:d(E', X^t) bound} that
	\[
	\sup_{\E'\in X^{t-\Delta'}}
	d(\E', X^{t})\underset{\Delta'\to0}{\longrightarrow}0
	\]
	It results in that there exists $\Delta>0$ such that for all $\tau\in[0, \Delta]$, 
	\[
	\sup_{\E'\in X^{t-\tau}}
	d(\E', X^{t})<\delta
	\]
	holds, that is, $d(\C, X^{t})<\delta$ holds for any $\C\in X^{t-\tau}$.
	Moreover, because $X^{t}$ is convex, there exists $\A\in X^{t}$ satisfying $d(\C, X^{t})=d(\C, \A)$, which proves the claim of the lemma.
\end{proof}
Note that a similar statement also holds for $Y^{t}$: there exists $\widetilde{\Delta}>0$ such that 
for all $\widetilde{\tau}\in[0, \widetilde{\Delta}]$ and for all $\D\in Y^{t-\widetilde{\tau}}$, there exists $\B\in Y^{t}$ satisfying $d(\D, \B)<\delta$.
Let $V:=O_{D}(2)\times O_{D}(2)(\simeq\mathcal{E}_{D}\times\mathcal{E}_{D})$ and 
let $d_{V}$ be a product metric on $V$ defined as 
\begin{align*}
	d_{V}\left((\A,\B), (\C,\D)\right) 
	= \max\{d(\A, \C), d(\B, \D)\}.
\end{align*}
According to Lemma \ref{lem: metric} and its $Y^{t}$-counterpart, if we take $\Delta_{0}=\min\{\Delta, \widetilde{\Delta}\}(>0)$, then 
there exists $(\A, \B)\in X^{t}\times Y^{t}$ for all $(\C, \D)\in X^{t-\Delta_{0}}\times Y^{t-\Delta_{0}}$ such that $d_{V}((\A,\B), (\C,\D))<\delta$.
On the other hand, as we have seen, it holds that
\[
X^{t}\times Y^{t} \cap JM(2:2) =\emptyset.
\]
Since $X^{t}\times Y^{t}$ and $JM(2:2)$ are closed in $V$, and $V$ is a metric space, we can apply Urysohn's Lemma \cite{GT75}.
It follows that there exists a continuous (in fact uniformly continuous since $V$ is compact) function $f\colon V\to [0,1]$ satisfying $f(U)=0$ for any $U\in X^{t}\times Y^{t}$ and $f(W)=1$ for any $W\in JM(2:2)$.
The uniform continuity of $f$ implies that for some $\varepsilon\in(0,1)$, there is $\delta>0$ such that
\begin{align}
	\label{eq:continuous f}
	\begin{aligned}
		&d_{V}\left((\E', \F'), (\E, \F)\right)<\delta\\
		&\qquad\qquad\Rightarrow
		\ 
		\left|f\left((\E', \F')\right)-f\left((\E, \F)\right)\right|<\varepsilon
	\end{aligned}
\end{align}
holds for any $(\E, \F)\in V$.
For this $\delta$, we can apply the argument above:
we can take $\Delta_{0}>0$ such that for any $(\C, \D)\in X^{t-\Delta_{0}}\times Y^{t-\Delta_{0}}$, there exists $(\A, \B)\in X^{t}\times Y^{t}$ satisfying $d_{V}((\A,\B), (\C,\D))<\delta$.
Because $f((\A,\B))=0$, we have $f((\C,\D))<\varepsilon<1$ (see \eqref{eq:continuous f}), and thus $(\C,\D)\notin JM(2:2)$ .
It indicates that $X^{t-\Delta_{0}}\times Y^{t-\Delta_{0}}\cap JM(2:2)=\emptyset$, that is, there is $\Delta_{0}>0$ for any $t\in L$ satisfying $t-\Delta_{0}\in L$.
Therefore, $t_{0}'=\inf L\notin L$ can be concluded.

\section{Conclusion}

In this study, we have introduced the notions of incompatibility and compatibility dimensions for collections of devices.
They describe the minimum number of states which are needed to detect incompatibility and the maximum number of states on which incompatibility vanishes, respectively.
We have not only presented general properties of those quantities but also examined concrete behaviors of them for a pair of unbiased qubit observables.
We have proved that even for this simple pair of incompatible observables there exist two types of incompatibility with different incompatibility dimensions which cannot be observed if we focus only on robustness of incompatibility under noise.
We expect that it is possible to apply this difference to some quantum protocols such as quantum cryptography.
Future work will be needed to investigate whether similar results can be obtained for observables in higher dimensional Hilbert space or other quantum devices.
As the definitions apply to devices in GPTs, an interesting task is further to see how quantum incompatibility dimension differ from incompatibility dimension in general.

\begin{acknowledgements}
The authors wish to thank Yui Kuramochi for helpful comments.
T.M. acknowledges financial support from JSPS KAKENHI Grant Number JP20K03732.
R.T. acknowledges financial support from JSPS KAKENHI Grant Number JP21J10096. 
\end{acknowledgements}

\end{document}